\documentclass[a4paper,12pt]{amsart}

\usepackage{amsmath}
\usepackage{amssymb}
\usepackage{amsfonts}
\usepackage{amsthm}
\usepackage{graphicx}

\DeclareMathAlphabet{\Ma}{U}{msa}{m}{n}
\DeclareMathAlphabet{\Mb}{U}{msb}{m}{n}

\DeclareMathAlphabet{\Meuf}{U}{euf}{m}{n}
\DeclareSymbolFont{ASMa}{U}{msa}{m}{n}
\DeclareSymbolFont{ASMb}{U}{msb}{m}{n}

\DeclareMathOperator{\ran}{ran}

\newcommand{\scalar}[2]{\langle#1\,,#2\rangle}
\newcommand{\scalarb}[2]{\langle#1\,,#2\rangle_{\partial \Omega}}
\newcommand{\scalarbw}[2]{\langle#1\,,#2\rangle_{W^\perp}}
\newcommand{\pair}[2]{(#1\,,#2)}

\newcommand{\norm}[1]{\|#1\|}
\newcommand{\normm}[1]{|\negthinspace\|#1\|\negthinspace|}

\newcommand{\A}{\mathcal{A}}
\renewcommand{\H}{\mathcal{H}}
\renewcommand{\L}{\mathcal{L}}

\newcommand{\D}{\mathcal{D}}
\renewcommand{\d}{\mathrm{d}}
\newcommand{\pO}{{\partial \Omega}}

\newcommand{\C}{\mathcal{C}}

\newcommand{\1}{\mathbb{I}}
\newcommand{\R}{\mathbb{R}}

\newtheorem{theorem}{Theorem}

\newtheorem{corollary}[theorem]{Corollary}
\newtheorem{proposition}[theorem]{Proposition}
\newtheorem{definition}[theorem]{Definition}
\newtheorem{lemma}[theorem]{Lemma}
\newtheorem{example}[theorem]{Example}
\newtheorem{remark}[theorem]{Remark}

\numberwithin{equation}{section}
\numberwithin{theorem}{section}

\title{
On self-adjoint extensions and symmetries in quantum mechanics
}

\date{}
\author{Alberto Ibort$^{1,2}$, Fernando Lled\'{o}$^{1,2}$ and Juan Manuel P\'erez-Pardo$^{1,2,3}$}
\address{$^1$Department of Mathematics,
University Carlos~III, Madrid, Avda. de la Universidad 30, E-28911 Legan\'es
(Madrid), Spain.}
\address{$^2$Instituto de Ciencias Matem\'{a}ticas (CSIC - UAM - UC3M - UCM), c./ Nicol\'{a}s Cabrera 13-15, 
Campus de Cantoblanco, UAM, 28049, Madrid, Spain.}
\address{$^3$ INFN-Sezione di Napoli, Via Cintia Edificio 6, I--80126 Napoli, Italy.
}

\email{albertoi@math.uc3m.es, flledo@math.uc3m.es, juanma@na.infn.it}

\subjclass[2010]{81Q10, 35J05, 46L60}
\keywords{Self-adjoint extensions of symmetric operators, quantum symmetries, reduction theory for unbounded operators, $G$-invariant 
Laplacians}

\thanks{The first and third name authors are partly supported by the project MTM2010-21186-C02-02
    of the spanish {\em Ministerio de Ciencia e Innovaci\'on} and QUITEMAD programme P2009 ESP-1594.
The second-named author was
partially supported by projects DGI MICIIN MTM2012-36372-C03-01 and Severo Ochoa SEV-2011-0087
of the spanish Ministry of Economy and Competition.
The third-named author was also partially supported in 2011 and 2012 by mobility grants of the
``\emph{Universidad Carlos III de Madrid}''
}

\begin{document}

\begin{abstract}
Given a unitary representation of a Lie group $G$ on a Hilbert space $\mathcal H$, we develop the theory of 
$G$-invariant self-adjoint extensions of symmetric operators both using von Neumann's theorem and the theory of quadratic forms. 
We also analyze the relation between the reduction theory of the unitary representation and the reduction of the 
$G$-invariant unbounded operator. We also prove a $G$-invariant version of the representation theorem for quadratic forms. 

The previous results are applied to the study of $G$-invariant
self-adjoint extensions of the Laplace-Beltrami operator on a smooth Riemannian manifold with boundary on which
the group $G$ acts. These extensions are labeled by admissible unitaries $U$ acting on the $L^2$-space at the boundary and having
spectral gap at $-1$.
It is shown that if the unitary representation $V$ of the symmetry group $G$ is traceable, then the self-adjoint extension of the 
Laplace-Beltrami operator determined by $U$ is $G$-invariant if $U$ and $V$ commute at the boundary. Various significant examples are discussed at the end.
\end{abstract}

\maketitle

\tableofcontents


\section{Introduction}

Symmetries of quantum mechanical systems are described by a group of transformations that preserves its essential structures. 
They play a fundamental role in studying the properties of the quantum system and reveal
fundamental aspects of the theory which are not present neither in the dynamics involved nor in the forces.
Space or time symmetries, internal symmetries, the study of invariant states or spontaneously broken symmetries
are standard ingredients in the description of quantum theories.
In many cases quantum numbers or superselection rules are labels characterizing representations of symmetry groups.
The publication of the seminal books of Weyl, Wigner and
van der Waerden (cf.~\cite{bWeyl28,bWigner31,bWaerden32}) in the late twenties
also indicates that quantum mechanics was using group theoretical methods
almost from its birth. We refer, e.g., to \cite[Chapter~12]{moretti13} or \cite{simon76}
for a more thorough introduction to various symmetry notions in quantum mechanics.

It was shown by Wigner that any symmetry transformation of a quantum system preserving the transition probabilities
between two states must be implemented by a semi-unitary (i.e., by a unitary or an anti-unitary) operator
(see, e.g., \cite[Introduction]{wigner39} or \cite[Chapters 2]{thaller92}).
The action of a symmetry group $G$ on a system is given in terms of a semi-unitary projective representation 
of $G$ on the physical Hilbert space, that can be described in terms of semi-unitary representations of $U(1)$-central extensions 
of the group or by means of an appropriate representation group (see, for instance, \cite{Bargmann,Santander75,Santander80}).   
Since the main examples of symmetries considered in this article
will be implemented in terms of unitary operators we will restrict here to this case. Moreover, anti-unitary
representation appear rarely in applications (typically implementing time reversal) and restrict to discrete
groups. The situation with an anti-unitary representation of a discrete symmetry group can also be easily incorporated
in our approach. 

In order to motivate how the symmetry can be implemented at the level of unbounded operators, consider
a self-adjoint Hamiltonian $T$ on the Hilbert space $\H$
and let $U(t):=e^{itT}$ be the strongly continuous one-parameter group implementing the unitary evolution
of the quantum system. Then, if $G$ is a quantum symmetry represented by the unitary representation
$V\colon G\to \mathcal{U}(\H)$ it is natural to require that $V$ and $U$ commute, i.e.,
\begin{equation}\label{commute}
 U(t) V(g) = V(g) U(t)\;,\quad t\in\R\;,\; g\in G\;.
\end{equation}
At the level of self-adjoint generators and, recalling that the domain of $T$ is given by
\[
 \D(T):=\left\{ \psi\in\H\mid \lim_{t\to 0}\frac{(U(t)-\1)\psi}{t}\quad\mathrm{exists}\;   \right\}\;,
\]
we have that (\ref{commute}) implies
\begin{equation}\label{eq:first-inv}
 V(g) \D(T) \subset \D(T)\quad\mathrm{and}\quad V(g)T\psi=TV(g)\psi\;,\quad \psi\in \D(T)\;.
\end{equation}
Of course, the requirement that the unitary representation $V$ of the symmetry group $G$ commutes with the 
dynamics of the system as in Eq.~(\ref{commute}) is restrictive. For example, if $V$ is a strongly
continuous representation of a Lie group, then (\ref{commute}) implies
the existence of conserved quantities that do not depend explicitly on 
time. Nevertheless, the previous comments justify that in the context of a single unbounded symmetric 
operator $T$ (not necessarily a Hamiltonian) it is reasonable to define $G$-invariance of $T$ 
as in Eq.~(\ref{eq:first-inv}) (see Section~\ref{sec:vonNeumannSymmetry} for details).

In the study of quantum systems it is standard
that some heuristic arguments suggest an expression for an observable
which is only symmetric on an initial dense domain but not self-adjoint.
The description of such systems is not complete until a self-adjoint extension of the operator has been determined,
e.g., a self-adjoint Hamiltonian operator $T$. Only in this case a unitary evolution of the system is given. This is due to the one-to-one
correspondence between densely defined self-adjoint operators and strongly continuous one-parameter groups of unitary operators
$U_t = \exp it T$ provided by Stone's theorem. 
The specification of a self-adjoint extension is typically done
by choosing suitable boundary conditions and this corresponds to a global understanding of the system
(see, e.g., \cite{ILP13,IbPer2010} and references therein).
Accordingly, the specification of the self-adjoint extension is not just
a mathematical technicality, but a crucial step in the description of the observables and the 
dynamics of the quantum system
(see, e.g., \cite[Chapter~X]{reed75} for further results and motivation).
We refer also to \cite{gitman12,moretti13,schmuedgen12} for recent textbooks that address systematically the problem of
self-adjoint extension from different points of view (see also the references therein).

The question of how does the process of selecting self-adjoint extensions of symmetric operators intertwine with the notion of quantum symmetry arises. This question is at the focus of our interest in this article.
We provide here natural characterizations of those self-adjoint extensions that are compatible
with the given symmetries. Concretely, if a symmetric operator is $G$-invariant in the sense of 
Eq.~(\ref{eq:first-inv}), then it is clear that 
not all self-adjoint extensions of the operator will also be $G$-invariant. 
This is evident if one fixes the self-adjoint extension by selecting boundary conditions. In general,
these conditions need not preserve the underlying symmetry of the system.
We present in Section~\ref{sec:vonNeumannSymmetry} and Section~\ref{sec:InvariantQF}
the characterization of $G$-invariant self-adjoint extensions from two different point of views: 
first, in the most general context of deficiency spaces provided by von Neumann's theorem.
Second, using the representation theorem of quadratic forms in terms of self-adjoint operators.
We prove in Theorem~\ref{QAinvariant} a $G$-invariant version of the representation theorem for quadratic forms.
In Section~\ref{sec:reduction} we give an alternative notion of $G$-invariance in terms of the 
theory of von Neumann algebras (cf., Proposition~\ref{pro:affiliate}). We relate also here
the irreducible sub-representations of $V$ with the reduction of 
the corresponding $G$-invariant self-adjoint extension $T$. In particular, we show that 
if $T$ is unbounded and $G$-invariant, then the group $G$ must act on the Hilbert space via a highly reducible representation $V$.
Finally, we apply the theory developed to a large class of self-adjoint extensions of the Laplace-Beltrami operator
on a smooth, compact manifold with smooth boundary on which a group is represented with a traceable unitary representation
(see Definition~\ref{def:traceable}).
In particular, self-adjoint extension of the Laplace-Beltrami operator with respect to groups acting by isometries 
on the manifold are discussed.
In this context the extensions are labeled by suitable unitaries on the boundary of
the manifold (see \cite{ILP13} for details). Concrete manifolds like a cylinder or a half-sphere
with $\mathbf{Z}_2$ or $SO(2)$ actions, respectively, will also be analyzed.

Apart from the previous considerations there are many instances where, though only partially, the previous problem has been considered.   
Just to mention a few here we refer to the analysis of translational symmetries and the study of self-adjoint 
extensions of the Laplacian in the description of a scalar quantum field in 1+1 dimensions in a cavity \cite{asorey06}.   
In a different vein we quote the spectral analysis of Hamiltonians in concentric spherical shells where the spherical symmetry is used 
in a critical way \cite{dittrich92,exner07}. In an operator theoretic context we refer, for example, to the notion of periodic Weyl-Titchmarsh
functions or invariant operators with respect to linear-fractional transformations
\cite{bekker04,bekker07}.
Even from a purely geometric viewpoint we should mention the analysis of isospectral 
manifolds in the presence of symmetries \cite{shams06}. We also refer to 
\cite[Section~13.5]{schmuedgen12} for the analysis of self-adjoint extensions commuting 
with a conjugation.

This article is organized as follows: in Section~\ref{sec:Preliminaries} we summarize well-known results on the theory of 
self-adjoint extensions, including the theory of scales of Hilbert spaces. In the next section
we introduce the main definitions concerning $G$-invariant operators and give an explicit
characterization of $G$-invariant self-adjoint extensions in the most general setting, i.e., using the 
abstract characterization due to von Neumann \cite{neumann30}. 
In Section~\ref{sec:InvariantQF} we introduce the notion of $G$-invariant quadratic forms and show
that the self-adjoint operators representing them will also be $G$-invariant operators. In the following section we present first steps of a 
reduction theory for $G$-invariant self-adjoint operators. For this we use 
systematically the notion of an unbounded operator affiliated to a von Neumann algebra.
In Section~\ref{sec:InvariantLB} we analyze the quadratic forms associated to the
Laplace-Beltrami operator when there is a Lie group acting on the manifold. Thus, we provide a characterization of the self-adjoint
extensions of the Laplace-Beltrami operator that are $G$-invariant.

{\bf Notation}: In this article all unbounded, linear operators $T$ that act on a separable, complex Hilbert space $\H$ are 
densely defined and we denote the corresponding domain by $\D(T)\subset\H$.

\section{Basic material on self-adjoint extensions}
\label{sec:Preliminaries}

For convenience of the reader and to fix our notation we will summarize here some standard facts on the theory of self-adjoint extensions
of symmetric operators, representation theorems for quadratic forms and the theory of rigged Hilbert spaces. We refer to standard references, e.g.,
\cite{reed80,akhiezer61b,schmuedgen12,kato95,koshmanenko99}, for proofs, further details and references.

\subsection{Symmetric and self-adjoint operators in Hilbert space}
\label{sec:vonNeumann}

Let $T$ be an unbounded, linear operator on the complex, separable Hilbert space $\H$ and with dense domain
$\D(T)\subset\H$. Recall that the operator $T$ is called symmetric if
$$\scalar{\Psi}{T\Phi}=\scalar{T\Psi}{\Phi}\quad\forall\Psi,\Phi\in\D(T)\;.$$
Moreover, $T$ is self-adjoint if it is symmetric and $\D(T)=\D(T^\dagger)$, where the
domain of the adjoint operator $\D(T^\dagger)$ is the set of all $\Psi\in\H$
such that there exists $\chi\in\H$ with
$$\scalar{\Psi}{T\Phi}=\scalar{\chi}{\Phi}\quad\forall\Phi\in\D(T)\;.$$
In this case we define $T^\dagger\Psi:=\chi$. If $T$ is symmetric then
$T^\dagger$ is a closed extension of $T$, $T\subset T^\dagger$, i.e., $\D(T)\subset\D(T^\dagger)$ and
$T^\dagger|_{\D(T)}=T\;.$

The relation between self-adjoint and closed, symmetric operators is subtle and extremely important, specially from the  physical point of view. It is thus natural to ask
if given a symmetric operator one can find a closed extension of it that is self-adjoint and
whether or not it is unique. Von Neumann addressed this issue in the late 20s and answered the question in an abstract setting,
cf., \cite{neumann30}. We recall the main definition and results needed later (see \cite[Theorem X.2]{reed75}).

\begin{definition}\label{def:deficiencyspaces}
Let $T$ be a closed, symmetric operator. The \textbf{deficiency spaces} $\mathcal{N}_{\pm}$ are defined to be
$$\mathcal{N}_{\pm}=\{\Phi\in\H\bigr|(T^\dagger\mp\mathbf{i})\Phi=0\}\;.$$ The \textbf{deficiency indices} are
$$n_{\pm}=\operatorname{dim}\mathcal{N}_{\pm}\;.$$
\end{definition}

\begin{theorem}[von Neumann]\label{thmvonNeumann}
Let $T$ be a closed, symmetric operator. The self-adjoint extensions of $T$ are in one-to-one correspondence with the set of unitaries
(in the usual inner product) of $\mathcal{N}_+$ onto $\mathcal{N}_-$. If $K$ is such a unitary then the corresponding self-adjoint operator
$T_K$ has domain $$\D(T_K)=\{\Phi+(\mathbb{I}+K)\xi\bigr|\Phi\in\D(T),\,\xi\in\mathcal{N}_+\}\;,$$ and
$$T_K\bigl(\Phi+(\mathbb{I}+K)\xi\bigr)=T^\dagger\bigl(\Phi+(\mathbb{I}+K)\xi\bigr)= T\Phi+\mathbf{i}(\mathbb{I}+K)\xi\;.$$
\end{theorem}

\begin{remark}
\begin{itemize}
\item[(i)]
The preceding definition and theorem can be also stated without assuming that the symmetric operator $T$
is closed (see, e.g., \cite[Section~XII.4]{DunfordII}). In view of Corollary~\ref{cor:Gclosure} and that in the context of
von Neumann algebras of Section~\ref{sec:reduction} the closure of $T$ is essential, we make this simplifying assumption
here.

\item[(ii)]
We refer to \cite{Posilicano08} for a recent article that characterizes the class of all self-adjoint extensions
of the symmetric operator obtained as a restriction of a self-adjoint operator to a suitable subspace of its domain. In particular,
the explicit relation of the techniques used to the classical result by von Neumann is also worked out.

\end{itemize}
\end{remark}

Finally, we recall that the densely defined operator $T\colon\D(T)\to\H$ is
\textbf{semi-bounded from below} if there is a constant $m\geq0$ such that
$$\scalar{\Phi}{T\Phi}\geq -m\norm{\Phi}^2\quad\forall\Phi\in\D(T)\;.$$
The operator $T$ is \textbf{positive} if the lower bound satisfies $m=0$\,.
Note that semi-bounded operators are automatically symmetric.

\subsection{Closable quadratic forms} \label{sec:closableqf}

In this section we introduce the notion of closed and closable quadratic forms. Standard references are,
e.g.,~\cite[Chapter~VI]{kato95}, \cite[Section~VIII.6]{reed80} or \cite[Section~4.4]{davies95}.

\begin{definition}
Let $\D$ be a dense subspace of the Hilbert space $\H$ and denote by $Q\colon\mathcal{D}\times\mathcal{D}\to \mathbb{C}$
a sesquilinear form (anti-linear in the first entry and linear in the second entry). The quadratic form associated to $Q$
with domain $\mathcal D$ is its evaluation on the diagonal, i.e., $Q(\Phi):=Q(\Phi,\Phi)$\,,
$\Phi\in\mathcal{D}$\,. The sesquilinear form is called \textbf{Hermitean} if
\[
  Q(\Phi,\Psi)=\overline{Q(\Psi,\Phi)}\;,\quad \Phi,\Psi\in\mathcal{D}\;.
\]
The quadratic form is  \textbf{semi-bounded from below} if there is an $m\geq 0$ such that
\[
  Q(\Phi)\geq -m \norm{\Phi}^2\;,\; \Phi \in \mathcal{D}\;.
\]
The smallest possible value $m$ satisfying the preceding inequality is called the \textbf{lower bound} for the quadratic form $Q$.
In particular, if  $Q(\Phi)\geq 0$ for all $\Phi\in\mathcal{D}$ then we call $Q$ \textbf{positive}.

\end{definition}

Note that if $Q$ is semi-bounded with lower bound $m$\,, then
$$Q(\Phi)+m\norm{\Phi}^2\;\,,\quad \Phi\in\D$$
is positive on the same domain.
We need to recall also the notions of closable and closed quadratic forms as
well as the fundamental representation theorems that relate closed, semi-bounded quadratic forms with self-adjoint, semi-bounded operators.

\begin{definition}
Let $Q$ be a semi-bounded quadratic form with lower bound $m\geq 0$ and dense domain $\D\subset\H$.
The quadratic form $Q$ is \textbf{closed} if $\D$ is closed with respect to the norm
\[
 \normm{\Phi}_Q:=\sqrt{Q(\Phi)+(1+m)\|\Phi\|^2}\;,\quad\Phi\in\D\;.
\]
If Q is closed and $\D_0\subset\D$ is dense with respect to the norm $\normm{\cdot}_Q$\,, then $\D_0$ is called a
\textbf{form core} for $Q$.

Conversely, the closed quadratic form $Q$ with domain $\D$ is called an
\textbf{extension} of the quadratic form $Q$ with domain $\D_0$. A quadratic form is said to be
\textbf{closable} if it has a closed extension.
\end{definition}

\begin{remark}$\phantom{=}$\label{Remclosable}
\begin{enumerate}
\item The norm $\normm{\cdot}_Q$ is induced by the following inner product on
the domain:
\[
 \langle\Phi,\Psi\rangle_Q:= Q(\Phi,\Psi)+(1+m)\langle\Phi,\Psi\rangle\;,\quad \Phi,\Psi\in\D\;.
\]
 \item It is always possible to close $\D \subset\mathcal{H}$ with respect to the norm $\normm{\cdot}_Q$.
The quadratic form is closable iff this closure is a subspace of $\H$.
 \end{enumerate}
\end{remark}

The following representation theorem shows the
deep relation between closed, semi-bounded quadratic forms and self-adjoint operators.
This result goes back to the pioneering work in the 50ies by Friedrichs, Kato, Lax and Milgram, and others 
(see, e.g., comments to Section VIII.6 in \cite{reed80}).
The representation theorem can be extended to the class of sectorial forms and operators
(see \cite[Section~VI.2]{kato95}), but we will only need here its version for
self-adjoint operators.

\begin{theorem}\label{fundteo}
Let $Q$ be an Hermitean, closed, semi-bounded quadratic form defined on the dense domain
$\D\subset\H$. Then it exists a unique, self-adjoint, semi-bounded operator $T$
with domain $\D(T)$ and the same lower bound such that
\begin{enumerate}
\item $\Psi\in\mathcal{D}(T)$ iff $\Psi\in \D$ and it exists $\chi \in \H$ such that
$$Q(\Phi,\Psi)=\langle\Phi,\chi\rangle\,,\quad\forall \Phi\in\D\;.$$
In this case we write $T\Psi=\chi$ and $Q(\Phi,\Psi)=\langle\Phi,T\Psi\rangle$ for any $\Phi\in\D$\,,\;$\Psi\in\D(T)$.
\item $\D(T)$ is a core for $Q$.
\end{enumerate}
\end{theorem}

Following \cite[Theorem~4.4.2]{davies95} we get the following characterization of
representable quadratic forms.
\begin{theorem}\label{pro:rep}
Let $Q$ be a semi-bounded, quadratic form with lower bound $m$ and domain $\D$. Then
the following conditions are equivalent:
\begin{enumerate}
 \item There is a lower semi-bounded operator $T$ with lower bound $m$ representing the quadratic form $Q$.
 \item The domain $\D$ is complete with respect to the norm $\normm{\cdot}_Q$.
\end{enumerate}
\end{theorem}

One of the most common uses of the representation theorem is to obtain self-adjoint extensions of symmetric, semi-bounded operators.
Given a semi-bounded, symmetric operator $T$ one can consider the associated quadratic form
$$Q_T(\Phi,\Psi)=\scalar{\Phi}{T\Psi}\quad \Phi,\Psi\in\D(T)\;.$$ These quadratic forms are always closable,
cf., \cite[Theorem X.23]{reed75}, and therefore their closure is associated to a unique self-adjoint operator.

Even if a symmetric operator has uncountably many possible self-adjoint extensions, the representation theorem above
allows to select a particular one given a suitable quadratic form. This extension is called the Friedrichs' or {\em hard}
extension and is in a natural sense maximal (see Chapters~10 and 13 in \cite{schmuedgen12} for a relation to Krein-von Neumann
or {\em soft} extensions).
The approach that we shall take in Section \ref{sec:InvariantQF} and Section \ref{sec:InvariantLB} uses this kind of
Friedrichs type extension.


\subsection{Scales of Hilbert spaces}
\label{sec:scalesHs}

The theory of scales of Hilbert spaces, also known as theory of rigged Hilbert spaces, has been used in
many ways in mathematics and mathematical physics. One of the standard applications of this theory appears in the
proof of the representation theorems mentioned above. We state next the main results,
(see \cite[Chapter I]{berezanskii68}, \cite[Chapter 2]{koshmanenko99} for proofs and more results).\\

Let $\H$ be a Hilbert space with scalar product $\scalar{\cdot}{\cdot}$ and induced norm $\norm{\cdot}$. Let $\H_+$ be a dense, linear
subspace of $\H$ which is a complete Hilbert space with respect to another scalar product that will be denoted by $\scalar{\cdot}{\cdot}_+$.
The corresponding norm is $\norm{\cdot}_+$ and we assume that
\begin{equation}\label{inclusion inequality}
\norm{\Phi}\leq\norm{\Phi}_+\;,\quad \Phi\in\H_+\;.
\end{equation}

Any vector $\Phi\in\H$ generates a continuous linear functional $L_\Phi\colon\H_+\to \mathbb{C}$  as follows. For $\Psi\in\H_+$ define
\begin{equation}
L_{\Phi}(\Psi)=\scalar{\Phi}{\Psi}\;.
\end{equation}
Continuity follows by the Cauchy-Schwartz inequality and Eq.~\eqref{inclusion inequality}.

Since $L_\Phi$ is a continuous linear functional on $\H_+$ it can be represented, according to Riesz theorem,
using the scalar product in $\H_+$. Namely, it exists  a vector $\xi\in\H_+$ such that
\begin{equation}
\forall\Psi\in\H_+\,,\quad L_{\Phi}(\Psi)=\scalar{\Phi}{\Psi}=\scalar{\xi}{\Psi}_+\;,
\end{equation}
and the norm of the functional coincides with the norm in $\H_+$ of the element $\xi$\,. 
One can use the above equalities to define an operator
\begin{equation}
\begin{array}{c}
\hat{I}\colon \H\to\H_+\\
\hat{I}\Phi=\xi\;.
\end{array}
\end{equation}
This operator is clearly injective since $\H_+$ is a dense subset of $\H$ and therefore it can be used to define a new scalar product on $\H$
\begin{equation}
\scalar{\cdot}{\cdot}_-:=\scalar{\hat{I}\cdot}{\hat{I}\cdot}_+\;.
\end{equation}
The completion of $\H$ with respect to this scalar product defines a new Hilbert space, $\H_-$\,, and the corresponding norm will be
denoted accordingly by $\norm{\cdot}_-$\,. It is clear that $\H_+\subset\H\subset \H_-$\,, with dense inclusions.
Since $\norm{\xi}_+=\norm{\hat{I}\Phi}_+=\norm{\Phi}_-$\,, the operator $\hat{I}$ can be extended by continuity to an isometric bijection.

\begin{definition}\label{def:scales}
The Hilbert spaces $\H_+$\,, $\H$ and $\H_-$ introduced above define a \textbf{scale of Hilbert spaces}. The extension by continuity of the
operator $\hat{I}$ is called the \textbf{canonical isometric bijection}. It is denoted by:
\begin{equation}
I\colon \H_-\to\H_+\;.
\end{equation}
\end{definition}

\begin{proposition}\label{proppairing}
The scalar product in $\H$ can be extended continuously to a pairing
\begin{equation}
\pair{\cdot}{\cdot}\colon\H_-\times\H_+\to\mathbb{C}\;.
\end{equation}
\end{proposition}

\begin{proof}
Let $\Phi\in\H$ and $\Psi\in\H_+$. Using the Cauchy-Schwartz inequality we have that
\begin{equation}\label{CSpairing}
|\scalar{\Phi}{\Psi}|=|\scalar{I\Phi}{\Psi}_+|\leq \norm{I\Phi}_+\norm{\Psi}_+=\norm{\Phi}_-\norm{\Psi}_+
\end{equation}
and we can extend the scalar product by continuity to the pairing $\pair{\cdot}{\cdot}$.
\end{proof}


\section{Self-adjoint extensions with symmetry}
\label{sec:vonNeumannSymmetry}

We begin now analyzing the question of how the process of finding a self-adjoint extension of a symmetric operator
intertwines with the notion of a quantum symmetry. We will denote by $G$ a group and let
\[
  V\colon G\to\mathcal{U}(\mathcal{H})
\]
be a fixed unitary representation of $G$ on the complex, separable Hilbert space $\mathcal H$. We will introduce the notion
of $G$-invariance by which we mean invariance under the fixed representation $V$.

\begin{definition}\label{def:G-inv}
Let $T$ be a linear operator with dense domain $\D(T)\subset\mathcal{H}$ and consider a unitary
representation $V\colon G\to\mathcal{U}(\H)$.
The operator $T$ is said to be \textbf{$G$-invariant} if
$TV(g)\supseteq V(g)T$,
i.e., if $V(g)\D(T)\subset\D(T)$ for all $g\in G$ and
$$T {V(g)}\Psi=V(g) T \Psi \quad \forall g\in G,\;\forall \Psi\in\D(T)\;.$$
\end{definition}

Due to the invertibility of the unitary operators representing the group we have the following
immediate consequence on $G$-invariant subspaces $\mathcal K$ of $\H$ which we will use several times:
\begin{equation} \label{subsetequal}
\mathrm{if}~V(g)\mathcal{K}\subset\mathcal{K}\,,\quad \forall g\in G\;,\;\;\mathrm{then~} V(g)\mathcal{K}=\mathcal{K}\quad \forall g\in G\,.
\end{equation}

\begin{proposition}\label{adjointinvariant}
Let $T\colon\D(T)\subset\H\to\H$ be a  G-invariant, symmetric operator. Then the adjoint operator $T^{\dagger}$ is $G$-invariant.
\end{proposition}

\begin{proof}
Let $\Psi\in\D(T^\dagger)$\,. Then, according to the definition of adjoint operator there is a
vector $\chi\in\H$ such that
$$\scalar{\Psi}{T\Phi}=\scalar{\chi}{\Phi}\quad \forall \Phi\in\D(T)\;.$$
Using the $G$-invariance we have
\begin{align*}
\scalar{V(g)\Psi}{T\Phi}&=\scalar{\Psi}{V(g^{-1})T\Phi}\\
&=\scalar{\Psi}{TV(g^{-1})\Phi}\\
&=\scalar{\chi}{V(g^{-1})\Phi}\\
&=\scalar{V(g)\chi}{\Phi}\;.
\end{align*}
The preceding equalities hold for any $\Phi\in\D(T)$ and therefore $V(g)\Psi\in\D(T^\dagger)$\,.
Moreover, we have that $T^\dagger V(g)\Psi=V(g)\chi=V(g)T^\dagger \Psi\;.$
\end{proof}

\begin{corollary}\label{cor:Gclosure}
 Let $T\colon\D(T)\subset\H\to\H$ be a $G$-invariant and symmetric operator on $\mathcal H$. Then its closure
 $\overline T$ is also $G$-invariant.
\end{corollary}
\begin{proof}
 The operator $T$ is symmetric and, therefore, closable. From $\overline{T}=T^{\dagger\dagger}$ and since
 $T$ is $G$-invariant, we have by the preceding proposition that $T^\dagger$ is $G$-invariant,
 hence also $\overline{T}=(T^\dagger)^\dagger$\,.
\end{proof}

The preceding result shows that we can always assume without loss of generality that the $G$-invariant
symmetric operators are closed.

We begin next with the analysis of the $G$-invariance of the self-adjoint extensions given by von Neumann's
classical result (cf., Theorem~\ref{thmvonNeumann}).

\begin{corollary}\label{cor:invariantdeficiency}
Let $T\colon\D(T)\subset\H\to\H$ be a closed, symmetric and $G$-invariant operator.
Then, the deficiency spaces $\mathcal{N}_{\pm}$, cf.,
Definition~\ref{def:deficiencyspaces}, are invariant under the action of the group, i.e.,
$$V(g)\mathcal{N}_{\pm}=\mathcal{N}_{\pm}\;.$$
\end{corollary}

\begin{proof}
Let $\xi\in\mathcal{N}_+\subset \D(T^\dagger)$\,. Then $(T^\dagger-i)\xi=0$ and we have 
from Proposition~\ref{adjointinvariant} that
\begin{equation*}
(T^\dagger-i)V(g)\xi=V(g)(T^\dagger-i)\xi=0\;.
\end{equation*}
This shows that $V(g)\mathcal{N}_{+}\subset \mathcal{N}_{+}$ for all $g \in G$
and by (\ref{subsetequal}) we get the equality. Similarly for $\mathcal{N}_-$\,.
\end{proof}

\begin{theorem}\label{Ginvariantoperator}
Let $T\colon\D(T)\subset\H\to\H$ be a closed, symmetric and $G$-invariant operator with equal deficiency indices
(cf., Definition~\ref{def:deficiencyspaces}).
Let $T_K$ be the self-adjoint extension of $T$ defined by the unitary $K\colon\mathcal{N}_+\to\mathcal{N}_-$\,.
Then $T_K$ is $G$-invariant iff $V(g)K \xi=K V(g) \xi$ for all $\xi\in\mathcal{N_+}$\,, $g\in G$\,.
\end{theorem}

\begin{proof}
To show the direction ``$\Leftarrow$'' recall that by
Theorem~\ref{thmvonNeumann} the domain of $T_K$ is given by
$$\D(T_K)=\D(T)\oplus\bigl(\mathbb{I}+K \bigr)\mathcal{N}_+\;.$$ Let $\Psi\in\D(T)$ and $\xi\in\mathcal{N}_+$\,. Then we have that
\begin{align*}
V(g)\bigl( \Psi + \bigl(\mathbb{I}+K\bigr)\xi \bigr)&=V(g)\Psi+ \bigl(V(g)+V(g)K\bigr)\xi\\
&=V(g)\Psi + \bigl(V(g)+KV(g)\bigr)\xi\\
&=V(g)\Psi+ \bigl(\mathbb{I}+K\bigr)V(g)\xi\;.
\end{align*}
By assumption $V(g)\Psi\in\D(T)$, and by Corollary \ref{cor:invariantdeficiency}, $V(g)\xi\in\mathcal{N_+}$\,. Hence $V(g)\D(T_K)\subset\D(T_K )$\,.
Moreover, we have that for $\Phi=\Psi+(\mathbb{I}+K)\xi\in\mathcal{D}(T_K)$
\begin{align*}
T_KV(g)\Phi&=T^\dagger V(g)\bigl(\Psi+\bigl(\mathbb{I}+K\bigr)\xi\bigr)\\
&=TV(g)\Psi+T^\dagger V(g)\bigl(\mathbb{I}+K\bigr)\xi\\
&=V(g)T\Psi + V(g)T^\dagger\bigl(\mathbb{I}+K\bigr)\xi=V(g)T_K\Phi\;,
\end{align*}
where we have used Proposition \ref{adjointinvariant}.

To prove the reverse implication ``$\Rightarrow$'' suppose that we have the self-adjoint extension defined by the unitary
$$K'=V(g)KV(g)^\dagger\;.$$
If we consider the domain $\D(T_{K'})$ defined by this unitary we have that
\begin{align*}
\D(T_{K'})&=\D(T) + \bigl(\mathbb{I}+V(g)KV(g)^\dagger\bigr)\mathcal{N}_+\\
&=V(g)\D(T) + V(g)\bigl(\mathbb{I}+K\bigr)V(g)^\dagger\mathcal{N}_+\\
&=V(g)\D(T_K)=\D(T_K)\;,
\end{align*}
where we have used again Proposition \ref{adjointinvariant}, Corollary~\ref{cor:invariantdeficiency} and 
(\ref{subsetequal}).
Now von Neumann's result stated in Theorem~\ref{thmvonNeumann} establishes a one-to-one correspondence between isometries
$K\colon\mathcal{N}_+\to\mathcal{N}_-$ and self-adjoint extensions of the operator $T$\,.
Therefore $K=K'=V(g)KV(g)^\dagger$ and the statement follows.
\end{proof}


\section{Invariant quadratic forms}
\label{sec:InvariantQF}

As mentioned in the first two sections the relation between closed, semi-bounded quadratic forms and self-adjoint operators
is realized through the so-called representation theorems.
We present here a notion of $G$-invariant quadratic form and prove a
representation theorem for $G$-invariant structures.

\begin{definition}
Let $Q$ be a quadratic form with domain $\D$ and let $V:G\to\mathcal{U}(\H)$ be a unitary representation of the group $G$\,.
We will say that the quadratic form is \textbf{$G$-invariant} if $V(g)\D\subset\D$  for all $g\in G$ and
$$Q(V(g)\Phi)=Q(\Phi)\quad\forall \Phi\in\D, \forall g\in G\;.$$
\end{definition}

It is clear by the polarization identity that if the quadratic form $Q$ is $G$-invariant,
then the associated sesquilinear form also satisfies $Q(V(g)\Phi,V(g)\Psi)=Q(\Phi,\Psi)$\,, $g\in G$.

We will now relate the notions of $G$-invariance for self-adjoint operators and for quadratic forms.

\begin{theorem}\label{QAinvariant}
Let $Q$ be a closed, semi-bounded quadratic form with domain $\D$ and let $T$ be the representing semi-bounded, self-adjoint operator.
The quadratic form $Q$ is $G$-invariant iff the operator $T$ is $G$-invariant.
\end{theorem}

\begin{proof}
To show the direction ``$\Rightarrow$'' recall from
Theorem~\ref{fundteo} that $\Psi\in\D(T)$ iff $\Psi\in\D$ and there exists $\chi\in\H$ such that
$$Q(\Phi,\Psi)=\scalar{\Phi}{\chi}\quad\forall \Phi\in\D\;.$$
Then, if $\Psi\in\D(T)$, and using the $G$-invariance of the quadratic form, we have that
\begin{align*}
Q(\Phi,V(g)\Psi)&=Q(V(g)^\dagger\Phi,\Psi)\\
                &=\scalar{V(g)^\dagger\Phi}{\chi}=\scalar{\Phi}{V(g)\chi}\;.
\end{align*}
This implies that $V(g)\Psi\in\D(T)$ and from
$$TV(g)\Psi=V(g)\chi=V(g)T\Psi\,,\quad \Psi\in\D(T)\,,g\in G\;,$$
we show the $G$-invariance of the self-adjoint operator $T$.

For the reverse implication ``$\Leftarrow$''
we use the fact that $\D(T)$ is a core for the quadratic form.
For $\Phi,\Psi\in\D(T)$ we have that
\begin{align*}
 Q(\Phi,\Psi)&= \scalar{\Phi}{T\Psi}=\scalar{V(g)\Phi}{V(g)T\Psi} \\
             &= \scalar{V(g)\Phi}{TV(g)\Psi}=Q(V(g)\Phi,V(g)\Psi)\;.
\end{align*}

These equalities show that the $G$-invariance of $Q$ is true at least for the elements in the domain of the operator.
Now for any $\Psi\in\D$ there is a sequence
$\{\Psi_n\}_n\subset\D(T)$ such that $\normm{\Psi_n-\Psi}_Q\to 0$\,. This, together with the equality above,
implies that $\{V(g)\Psi_n\}_n$ is a Cauchy sequence with respect to $\normm{\cdot}_{Q}$\,. Since $Q$ is closed,
the limit of this sequence is in $\D$\,. Moreover it is clear that $\lim_{n\to\infty}V(g)\Psi_n=V(g)\Psi$\,,
hence
$$\normm{V(g)\Psi_n-V(g)\Psi}_Q\to0\,.$$

So far we have proved that $V(g)\D\subset\D$. Now for any $\Phi,\Psi\in\D$ consider
sequences $\{\Phi_n\}_n,\{\Psi_n\}_n\subset\D(T)$ that converge respectively to $\Phi,\Psi\in\D$ in the topology induced by $\normm{\cdot}_{Q}$\,. Then the limit
\begin{align*}
Q(\Phi,\Psi)&=\lim_{n\to\infty}\lim_{m\to\infty}Q(\Phi_n,\Psi_m)\\
            &=\lim_{n\to\infty}\lim_{m\to\infty}Q(V(g)\Phi_n,V(g)\Psi_m)=Q(V(g)\Phi,V(g)\Psi)\;.
\end{align*}
concludes the proof.
\end{proof}

The preceding result and Theorem~\ref{pro:rep} allow to give the following characterization of
representable $G$-invariant quadratic forms.

\begin{proposition}\label{InvariantRepresentation}
Let $Q$ be a $G$-invariant, semi-bounded quadratic form with lower bound $m$ and domain $\D$. The following statements are equivalent:
\begin{enumerate}
\item[(i)] There is a $G$-invariant, self-adjoint operator $T$ on $\D(T)\subset\H$ with lower semi-bound $m$ and that represents the quadratic form,
i.e., $$Q(\Phi,\Psi)=\scalar{\Phi}{T\Psi}\quad\forall \Phi\in\D,\forall \Psi\in \D(T)\;.$$
\item[(ii)] The domain $\D$ of the quadratic form is complete in the norm $\normm{\cdot}_{Q}$.
\end{enumerate}
\end{proposition}

To conclude this section we make contact with the theory of scales of Hilbert spaces introduced in Section~\ref{sec:scalesHs}.
Let $Q$ be a closed, semi-bounded quadratic form with domain $\D\subset\H$. We will show that if $Q$ is $G$-invariant then one
can automatically produce unitary representations $V_\pm$ on the natural scale of Hilbert spaces
$$
 \mathcal{H}_+\subset \mathcal{H} \subset\mathcal{H}_-\;,
$$
where $\mathcal{H}_+:=\D$.

\begin{theorem}\label{thm:unit-scales}
Let $Q$ be a closed, semi-bounded, $G$-invariant quadratic form with lower bound $m$. Then
\begin{enumerate}
\item $V$ restricts to a unitary representation $V_+$ on $\mathcal{H}_+:=\mathcal{D}\subset\mathcal{H}$ with scalar product given by
$$\scalar{\Phi}{\Psi}_+:=\scalar{\Phi}{\Psi}_Q=(1+m)\scalar{\Phi}{\Psi}+Q(\Phi,\Psi)\,,\quad \Phi,\Psi\in\H_+\;.$$

\item $V$ extends to a unitary representation $V_-$ on $\mathcal{H}_-$ and we have, on $\mathcal{H}_-$\,,
\begin{equation}\label{eq:commute-I}
V_+(g) I=I V_-(g)\;,\quad g\in G\;,
\end{equation}
where $I\colon \mathcal{H}_-\to \mathcal{H}_+$ is the canonical isometric bijection of Definition \ref{def:scales}.
	\end{enumerate}
\end{theorem}
\begin{proof}
(i) To show that the representation $V_+:=V|_{\H_+}$ is unitary with respect to $\scalar{\cdot}{\cdot}_+$
    note that by definition of $G$-invariance of the quadratic form we have
    for any $g\in G$ that $V(g)\colon\mathcal{H}_+\to \mathcal{H}_+$ and
$$
 \scalar{V(g)\Phi}{V(g)\Psi}_+=\scalar{\Phi}{\Psi}_+\quad \Phi,\Psi\in\H_+\,.
$$
Since any $V(g)$ is invertible we conclude that $V$ restricts to a unitary representation on $\mathcal{H}_+$\,.

(ii) To show that $V$ extends to a unitary representation $V_-$ on $\H_-$ consider first the following representation of $G$ on $\H_-$\,:
$$
V_-(g):=I^{-1}V_+(g)I\;.
$$
We show first that this representation is unitary: since $V_-$ is invertible it is enough to check the
isometry condition using part (i). Indeed, for any $\alpha,\beta\in\H_-$ we have
\begin{align*}
\scalar{V_-(g)\alpha}{V_-(g)\beta}_-&=\scalar{I^{-1}V_+(g)I\alpha}{I^{-1}V_+(g)I\beta}_-\\
&=\scalar{V_+(g)I\alpha}{V_+(g)I\beta}_+=\scalar{I\alpha}{I\beta}_+=\scalar{\alpha}{\beta}_-\;.
\end{align*}
The restriction of $V_-(g)$, $g\in G$, to $\H$ coincides with $V(g)$. Indeed, consider the pairing $\pair{\cdot}{\cdot}\colon\H_-\times\H_+\to\mathbb{C}$ of Proposition~\ref{proppairing} and 
let $\Phi\in\H\subset\H_-$. Then for all $\Psi\in\H_+$
\begin{align*}
\pair{V_-(g)\Phi}{\Psi}&=\pair{I^{-1}V_+(g)I\Phi}{\Psi}\\
&=\scalar{V(g)_+I\Phi}{\Psi}_+\\
&=\scalar{I\Phi}{V_+(g^{-1})\Psi}_+\\
&=\pair{\Phi}{V_+(g^{-1})\Psi}\\
&=\scalar{\Phi}{V(g^{-1})\Psi}\\
&=\scalar{V(g)\Phi}{\Psi}=\pair{V(g)\Phi}{\Psi}\;.
\end{align*}
Since $V_-(g)$ is a bounded operator in $\H_-$ and $\H$ is dense in $\H_-$, $V_-(g)$ is the extension of $V(g)$ to $\H_-$\,.
\end{proof}

\section{Reduction theory}\label{sec:reduction}

The aim of this section is to provide an alternative point of view for the notion of $G$-invariance of operators
(cf., Section~\ref{sec:vonNeumannSymmetry}) in terms of von Neumann algebras. Based on this approach we will
address some reduction issues of the unbounded operator in terms of the reducibility of unitary representation $V$
implementing the quantum symmetry.

Recall that a von Neumann algebra $\A$ is a unital *-subalgebra of $\mathcal{L}(\H)$
(the set of bounded linear operators in $\H$) which is closed in the weak operator topology.
Even if a von Neumann algebra consists only of bounded linear operators acting on a
Hilbert space, this class of operator algebras can be related in a natural way to closed, unbounded and densely defined
operators. In fact, von Neumann introduced these algebras in 1929 and proved the celebrated bicommutant theorem,
when he extended the spectral theorem to closed, unbounded normal operators in a
Hilbert space (cf., \cite{vNeumann29}). Since then, the 
notion of affiliation of an unbounded operator to an operator algebra has been applied in different situations
(see, e.g., \cite{Woronowicz91,Woronowicz92,BY90}).

Let $\mathcal{S}$ be a self-adjoint subset of
$\mathcal{L}(\H)$, i.e., if $S\in\mathcal{S}$, then $S^*\in\mathcal{S}$. 
We denote by $\mathcal{S}'$ the commutant of $\mathcal{S}$
in $\mathcal{L}(\H)$, i.e., the set of all bounded and linear operators on
$\H$ commuting with all operators in $\mathcal{S}$. It is a fact that
$\mathcal{S}'$ is a von Neumann algebra and that the corresponding
bicommutant $\mathcal{S}'':=(\mathcal{S}')'$ is the smallest von Neumann algebra
containing $\mathcal{S}$, i.e., $\mathcal{S}''$ is the von Neumann algebra generated by
the set $\mathcal{S}\subset\mathcal{L}(\H)$. We refer to Sections~4.6 and 5.2 of
\cite{Pedersen89} for further details and proofs.

The definition of commutant of a densely defined unbounded operator $T$ is more delicate since
one has to take into account the domains. The following definition generalizes the notion of commutant
mentioned before.

\begin{definition}\label{def:com}
Let $T\colon\mathcal{D}(T)\subset\mathcal{H}\to \H$ be a closed, densely defined operator. The commutant of $T$
is given by 
\[
 \{T\}':=\{A\in\mathcal{L}(\mathcal{H})\mid A\mathcal{D}(T)\subset \mathcal{D}(T)\quad\mathrm{and}\quad
           TA\Phi=AT\Phi\;,\; \Phi\in\mathcal{D}(T)\}\;,
\]
i.e., $A\in\{T\}'$ if $TA\supseteq AT$.
\end{definition}

Since $T$ is a closed operator we have that $\{T\}'\cap \{T^\dagger\}'$ is a von Neumann algebra in $\mathcal{L}(\mathcal{H})$.
We denote the von Neumann algebra associated to the bounded components of $T$ as
\[
 W^*(T):=\left(\{T\}'\cap \{T^\dagger\}'\right)'\subset\mathcal{L}(\mathcal{H})\;.
\]
In particular, if $T$ is self-adjoint, the spectral projections of $T$ are contained in 
$W^*(T)$.
\begin{definition}\label{def:affi}
The closed, densely defined operator $T\colon\mathcal{D}(T)\subset\mathcal{H}\to \H$ is {\bf affiliated to a von Neumann algebra}
$\mathcal A\subset\mathcal{L}(\mathcal{H})$ (and we write this as $T a \mathcal A$) if
\[
 W^*(T)\subset\mathcal{A}\;.
\]
\end{definition}

\begin{remark}
There are different equivalent characterizations of the notion of affiliation: 
$Ta \mathcal A$ iff $\{T\}'\cap \{T^\dagger\}'\supset\mathcal{A}'$. In particular, this implies that
$TA'\supseteq A'T$, $A'\in\mathcal{A}'$, i.e., $A' \mathcal{D}(T)\subset \mathcal{D}(T)$ and 
$TA'\Phi= A'T\Phi$ for all $A'\in\mathcal{A}'$, $\Phi\in\mathcal{D}(T)$.
If $T$ is an (unbounded) self-adjoint operator, then
$W^*(T)$ coincides with the von Neumann algebra $\mathcal{C}(T)''$ generated by the Cayley transform of $T$. 
Recall that the Cayley transform
\[
 \mathcal{C}(T):=(\1-i T)(\1+iT)^{-1}\in\mathcal{U}(\mathcal{H})
\]
is a unitary that can be associated with $T$. We conclude that $T a \mathcal{A}$ iff $\mathcal{C}(T)\in\mathcal{A}$.

Finally, note that if $T$ is a bounded operator then, $W^*(T)=\{T,T^\dagger\}''$ is the von Neumann algebra generated by
$T, T^\dagger$ and $T a \mathcal A$ iff $T, T^\dagger\in\mathcal{A}$. 
\end{remark}

In the following result we will give a useful characterization of $G$-invariance for symmetric operators
in terms of the affiliation to the commutant of the quantum symmetry.

\begin{proposition}\label{pro:affiliate}
Let $T\colon\mathcal{D}(T)\subset\mathcal{H}\to \H$ be a closed, symmetric operator.
Then, $T$ is $G$-invariant iff $T a \mathcal{V}'$,
where $\mathcal V$ is the von Neumann algebra generated by $\{V(g)\mid g\in G\}$
(i.e.,
$
 \mathcal{V}=\{V(g)\mid g\in G\}''\subset\mathcal{L}(\mathcal{H})
$) and $\mathcal{V}'$ its commutant. Moreover, any $G$-invariant
self-adjoint extension of $T$ is also affiliated to $\mathcal{V}'$.
\end{proposition}
\begin{proof}

If $T a \mathcal{V}'$, then it is immediate that $T$ is $G$-invariant, since 
$(\mathcal{V}')'\subset\{T\}'\cap \{T^\dagger\}'$ and therefore the 
generators $\{V(g)\mid g\in G\}$ of the von Neumann algebra $\mathcal{V}=\mathcal{V}''$
satisfy $\{V(g)\mid g\in G\}\subset\{T\}'$. This gives the $G$-invariance of $T$ (cf.,
Definition~\ref{def:com}).

To show the reverse implication assume that $T$ is $G$-invariant according to 
Definition~\ref{def:G-inv}, i.e., $\{V(g)\mid g\in G\}\subset \{T\}'$.
From Proposition~\ref{adjointinvariant} we also have that
\[
 \{V(g)\mid g\in G\}\subset\{T\}'\cap \{T^\dagger\}'\;,
\]
hence
\[
 \mathcal{V}'=\{V(g)\mid g\in G\}'\supset\left(\{T\}'\cap \{T^\dagger\}'\right)'=W^*(T)\;,
\]
which implies that $T a \mathcal{V}'$. The same argument shows that any $G$-invariant, self-adjoint
extension of $T$ is also affiliated to $\mathcal{V}'$.
\end{proof}

We begin next with the analysis of the relation between the reducing subspaces of the 
quantum symmetry $V$ and those of the self-adjoint operator $T$ defined on the dense domain
$\mathcal{D}(T)$. The following result and part~(ii) of Theorem~\ref{thm:disjointirrep} are 
straightforward consequences of Schur's lemma.

\begin{lemma}\label{lem:irrep}
Let $T\colon\mathcal{D}(T)\subset\mathcal{H}\to \H$ be a self-adjoint operator. If $T$ is 
$G$-invariant with respect to a unitary, irreducible representation $V$ of $G$ on the Hilbert
space $\mathcal{H}$, then $T$ must be bounded and
\[
 T=i\left(\frac{\lambda-1}{\lambda+1}\right)\1\quad\mathrm{for~some}\quad \lambda\in\mathbb{T}\setminus \{-1\}\;.
\]
\end{lemma}
\begin{proof}
Schur's lemma and the irreducibility of $V$ imply that $\mathcal{V}'=\mathbb{C}\1$. 
Moreover, by Proposition~\ref{pro:affiliate} and since the Cayley transform is a unitary and
$\mathcal{C}(T)\in\mathcal{V}'$ we have
\[
 \mathcal{C}(T)=(\1-iT)(\1+iT)^{-1}=\lambda\1\quad\textrm{for~some}\quad \lambda\in\mathbb{T}\;.
\]
The case $\lambda=-1$ is not possible since $\mathcal{C}(T)$ is an isometry of $(\1+iT)\mathcal{D}(T)$
onto $(\1-iT)\mathcal{D}(T)$ and $\mathcal{D}(T)$ is dense. Therefore $\mathcal{C}(T)=\lambda\1$
for some $\lambda\in\mathbb{T}\setminus \{-1\}$.
This implies that for any $\Phi\in\mathcal{D}(T)$ we have
\[
 T\Phi=i\left(\frac{\lambda-1}{\lambda+1}\right) \Phi\;.
\]
Since the right hand side of the previous equation is bounded we can extend
the formula for $T$ to the whole Hilbert space.
\end{proof}

To continue our analysis we have to define first in which sense an unbounded operator can 
be reduced by a closed subspace.
Roughly speaking, the reduction means that we can write $T$ as the sum of two
parts: one acting on the reducing subspace and one acting on its orthogonal complement.
The following definition generalizes the standard one for bounded operators and uses the 
notion of commutant of an unbounded operator as in Definition~\ref{def:com}.

\begin{definition}
Let $T\colon\mathcal{D}(T)\subset\mathcal{H}\to \H$ be a self-adjoint operator and
$\H_1$ be a closed subspace of the Hilbert space $\H$. We denote by $P_1$ the orthogonal projection
onto $\H_1$. The subspace $\H_1$ (or $P_1$) reduces $T$ if $P_1\in\{T\}'$, i.e., if
$P_1\D(T)\subset\D(T)$ and $TP_1\,\Phi= P_1T\,\Phi$, $\Phi\in\D(T)$.
\end{definition}

The previous definition implies that if $\H_1$ is reducing for $T$, then $P_1^\perp=\1-P_1$ is also 
reducing and $\D(T)=\D(T)\cap\H_1+\D(T)\cap\H_1^\perp$. Moreover, the subspace $\H_1$ is invariant in the sense that
\[
 T|(\D(T)\cap\H_1)\subset \H_1\quad\mathrm{and}\quad T|(\D(T)\cap\H_1^\perp)\subset \H_1^\perp\;.
\]
If $T$ is self-adjoint, then the spectral projections $E(\omega)$ (with $\omega$ Borel on 
the spectrum $\sigma(T)$) reduce $T$.

\begin{theorem}\label{thm:disjointirrep}
Let $G$ be a group and consider a unitary, reducible representation $V$ which decomposes as 
\[
 V=\mathop{\oplus}_{n=1}^N V_n \quad\mathrm{on}\quad \H=\mathop{\oplus}_{n=1}^N \H_n\;,\quad N\in\mathbb{N}\cup \{\infty\}\;,
\]
where the sub-representations $V_k$, $k=1,\dots, N$ are irreducible and mutually inequivalent.
Let $T\colon\mathcal{D}(T)\subset\mathcal{H}\to \H$ be a self-adjoint and $G$-invariant operator with respect to
the representation $V$. Then
\begin{itemize}
 \item[(i)] Any projection $P_k$ onto $\H_k$, $k=1,\dots, N$, is central in $\mathcal V$, (i.e., $P_k\in\mathcal{V}\cap\mathcal{V}'$), 
            and reduces the operator $T$, (i.e., $P_k\in\{T\}'$).
 \item[(ii)] If $N<\infty$, then $T$ must be a bounded operator and there exist $\lambda_k\in\mathbb{T}\setminus \{-1\}$, $k=1,\dots, N$, such that
\[
 T\cong i \;\mathrm{diag}\left( 
       \left(\frac{\lambda_1-1}{\lambda_1+1}\right)\1_{\mathcal{H}_{1}}, \dots, \left(\frac{\lambda_N-1}{\lambda_N+1}\right)\1_{\mathcal{H}_{N}}
                          \right)\;.
\]

\end{itemize}
\end{theorem}
\begin{proof}
(i) Since the $V_k$'s are all irreducible and mutually inequivalent it follows by Schur's lemma that
\[
 \mathcal{V}'\cong 
         \Big\{
         \mathrm{diag}\left( 
         \lambda_1\1_{\mathcal{H}_{1}}, \dots, \lambda_N\1_{\mathcal{H}_{N}}
                     \right)
         \mid \lambda_1,\dots,\lambda_N\in\mathbb{C}
         \Big\}\;.
\]
Moreover, since any $P_k$ reduces $V$ it is immediate that the projections are central. To show that
$P_k\in\{T\}'$, $k=1,\dots, N$, consider the spectral projections $E(\cdot)$ of $T$ and define for 
any $\Phi \in\mathcal{D}(T)$ the following positive finite measure on the Borel sets of $\sigma(T)$:
\[
 \mu_\Phi (\omega):=\|E(\omega)\Phi \|^2=\langle \Phi , E(\omega)\Phi \rangle\;.
\]
Now, any $P_k$ is central for $\mathcal V$ and, by $G$-invariance, Proposition~\ref{pro:affiliate} 
implies that $E(\omega)\in\mathcal{V}'$ for any Borel set $\omega$. Therefore we have
\[
 \mu_{P_k \Phi }(\omega) =\|E(\omega) P_k \Phi \|^2\leq \|E(\omega)\Phi \|^2=\mu_\Phi (\omega)
\]
and this implies that $P_k\mathcal{D}(T)\subset\mathcal{D}(T)$. Similarly, using the spectral theorem
one can show that
\[
 \langle y, TP_k\, \Phi \rangle =\langle y, P_k\, T \,\Phi \rangle\quad\mathrm{for~all}\quad y\in\H\;,\;\Phi \in\mathcal{D}(T)\;,
\]
hence $P_k\in\{T\}'$. 

(ii) Since $T=T^*$ we have that the Cayley transform is unitary and 
\[   
\mathcal{C}(T)=(\1-iT)(\1+iT)^{-1}\in\mathcal{V}'\;.
\]
Therefore, there is a $\lambda_k\in\mathbb{T}$, $k=1,\dots, N$, such that
\[
\mathcal{C}(T)\cong
\mathrm{diag}\left( 
         \lambda_1\1_{\mathcal{H}_{1}}, \dots, \lambda_N\1_{\mathcal{H}_{N}}
                     \right)\;.
\]
As in the proof of Lemma~\ref{lem:irrep} we exclude first the case $\lambda_k=-1$, $k=1,\dots, N$.
If $\lambda_k=-1$ and since the projection $P_k$ is reducing we have for any $\Phi _k\in P_k\mathcal{D}(T)$
that
\[
 (\1+iT)\Phi _k = -(\1-iT)\Phi _k \quad\Rightarrow\quad \Phi _k=0\;.
\]
Therefore $P_k\mathcal{D}(T)=\{0\}$ and we can omit the $k$th-summand in the decomposition of $T$.
Hence without loss of generality we can assume that $\lambda_k\in\mathbb{T}\setminus \{-1\}$,
$k=1,\dots, N$ and a similar reasoning on each block as in Lemma~\ref{lem:irrep} gives the formula for $T$.
\end{proof}

Part~(ii) of the previous theorem says that any representation of $V$ implementing a 
quantum symmetry of an unbounded, self-adjoint operator must be highly reducible. 
Note that only if $N=\infty$ may $T$ be unbounded. E.g., consider the case
where $\lim_N\lambda_N= -1$.
In the particular case of a compact group acting on an infinite dimensional
Hilbert space, we know that the decomposition of $V$ into irreducible representations must have infinite irreducible 
components.
In this sense the representations considered in the examples of the following sections are meaningful. 
The following remark and Proposition~\ref{prop:equivalentirrep} show that this is so even if we consider equivalent irreducible representations.

\begin{remark}
If the irreducible representations are not mutually inequivalent, then the corresponding projections need not be reducing.
In fact, consider the example $V=V_1\oplus V_2$ on $\H=\H_1\oplus \H_2$ with $V_1\cong V_2$, i.e., there is a unitary
$U\colon\H_1\to\H_2$ such that $V_2=UV_1U^*$. Then 
\[
 \mathcal{V}'=\left\{
                \begin{pmatrix}
                  \lambda_1 \1_{\mathcal{H}_1} &  \lambda_2 U^* \\
                  \lambda_3 U  & \lambda_4 \1_{\mathcal{H}_2} 
                \end{pmatrix}
                 \mid \lambda_k\in\mathbb{C}, k=1,\dots,4
              \right\}
\] and 
\[
\mathcal{V}=\mathcal{V}''=
               \left\{
                \begin{pmatrix}
                  A_1 &  0 \\
                  0   & UA_1U^* 
                \end{pmatrix}
                 \mid A_1\in\mathcal{L}(\H_1)
              \right\}\;.
\]
This shows that $P_k\notin\mathcal{V}$ and, in fact if $V$ is a quantum symmetry for $T$, then $P_k$ need not be reducing for $T$.
\end{remark}

It can be shown that $T$ must also be bounded in this later case. Take into account that it is not assumed that the irreducible 
representations are finite dimensional. Below we show this in the simple case that the representation $V$ is a composition of 
two equivalent representations. The generalization to a finite number of equivalent representations is straightforward.

\begin{proposition}\label{prop:equivalentirrep}
Let $G$ be a group and consider a unitary, reducible representation $V$ which decomposes as a direct sum of two equivalent, irreducible representations. 
Let $T\colon \D(T)\subset\H\to\H$ be a $G$-invariant, self-adjoint operator with respect to the representation $V$. Then $T$ 
must be a bounded operator.
\end{proposition}

\begin{proof}
By assumption we have that $V=V_1\oplus V_2$ with $V_2= UV_1U^*$ where $U:\H_1\to\H_2$ is the unitary operator 
representing the equivalence and $\H=\H_1\oplus \H_2$. According to the previous remark we have that the Cayley transform of the operator $T$ is 
$$\mathcal{C}(T)=\begin{pmatrix}
                  \lambda_1 \1_{\H_1} &  \lambda_2 U^* \\
                  \lambda_3 U  & \lambda_4 \1_{\H_2} 
                \end{pmatrix}\;,\quad\mathrm{for~some}\quad\lambda_k\in\mathbb{C}, k=1,\dots, 4\;.$$
Moreover, since $\C(T)$ is a unitary operator, the coefficient matrix 
$$\Lambda:=	\begin{pmatrix}
\lambda_1 & \lambda_2 \\ \lambda_ 3 & \lambda_4
\end{pmatrix}$$
is a $2\times 2$ unitary matrix, i.e., $\Lambda\in\mathcal{U}(2)$\,. Therefore it exists a unique, unitary matrix 
$$\Sigma=	\begin{pmatrix}
s_1 & s_2 \\ s_ 3 & s_4
\end{pmatrix}\in \mathcal{U}(2)$$
that diagonalizes $\Lambda$, i.e., $\Sigma^*\Lambda \Sigma=\mathrm{diag}(\tilde{\lambda}_1,\tilde{\lambda}_2)$\,, 
$\tilde{\lambda_k}\in\mathbb{T},k=1,2$\,. Consider the unitary operator 
$$S=	\begin{pmatrix}
s_1\1_{\H_1} & s_2U^* \\ s_3U & s_4\1_{\H_2}
\end{pmatrix}\;.$$
This unitary transformation satisfies that 
\begin{equation}
S^*\C(T)S=\begin{pmatrix} \tilde{\lambda}_1 \,\1_{\tilde{\H}_1}& 0 \\ 
                          0 & \tilde{\lambda}_ 2\, \1_{\tilde{\H}_2} \end{pmatrix}
                          \quad\mathrm{for~some}\quad\tilde{\lambda}_1\,,\,\tilde{\lambda}_2\in\mathbb{T} \;,
\end{equation}
where the new block structure represents a different decomposition of $\H=\tilde{\H}_1\oplus\tilde{\H}_2$\,. 
With respect to this decomposition there are associated two proper subspaces of $\C(T)$ with proper projections $\tilde{P}_1$ and $\tilde{P}_2$. 
Notice that these projections reduce $T$. These projections do not coincide in general with $P_1$ and $P_2$\,, 
the projections associated to the decomposition $V=V_1\oplus V_2$\,.

The same arguments as in the proof of Theorem \ref{thm:disjointirrep} lead us to exclude the cases $\tilde{\lambda}_k=-1$ 
and we can consider that $\tilde{\lambda}_k\in\mathbb{T}\setminus \{-1\}, k=1,2$\,. Hence $T$ is a bounded operator.
\end{proof}

The following result is a straightforward consequence of the main results in this section.

\begin{corollary}
Let $T$ be an unbounded self-adjoint operator on a Hilbert space $\mathcal H$ which is $G$-invariant with respect to 
a unitary representation $V\colon G\to\mathcal{U}(\mathcal{H})$. Then $V$ cannot be a direct sum of finitely many irreducible 
representations.
\end{corollary}

The present section is a first step to analyze the relation between the reduction theory of a quantum mechanical symmetry
and the reduction of the unbounded $G$-invariant operator. 
In particular we consider only when $V$ decomposes as a direct
sum of irreducible representations. This is enough for the applications we have in mind in the following sections, where mainly 
compact groups act as a quantum symmetry. For a systematic and general theory of reduction one has to address, among other things,
the type decomposition of the von Neumann algebras corresponding to the intertwiner spaces of the representation $V$ and 
the corresponding direct integral decomposition of the self-adjoint operator
$T$ (see, e.g., \cite{Nussbaum64,Mackey76}).


\section{Invariant self-adjoint extensions of the Laplace-Beltrami operator}
\label{sec:InvariantLB}

As an application of the previous results we analyze the class of self-adjoint extensions
of the Laplace-Beltrami operator on a Riemannian manifold introduced in \cite{ILP13} according to their
invariance properties with respect to a symmetry group, in particular with respect to a group action on the manifold. 

Throughout the rest of this and the next section we will consider a smooth, compact, Riemannian manifold with boundary $(\Omega,\pO,\eta)$\,. 
The boundary $\pO$ of the Riemannian manifold $(\Omega,\pO,\eta)$ has itself the structure of a Riemannian manifold without 
boundary $(\pO,\partial\eta)$\,.  The Riemannian metric at the boundary is just the pull-back of the Riemannian metric $\partial\eta=i^*\eta$\,, 
where $i:\pO\hookrightarrow\Omega$ is the canonical inclusion map. The spaces of smooth functions over the manifolds verify that 
$$\C^\infty(\Omega)\bigr|_{\pO}\simeq\C^\infty(\pO)\;.$$
The Sobolev spaces of order $k\in\mathbb{R}^+$ over the manifolds $(\Omega,\pO,\eta)$ and $(\pO,\partial\eta)$ are going 
to be denoted by $\H^k(\Omega)$ and $\H^k(\pO)$, respectively. There is an important relation between the Sobolev spaces 
defined over the manifolds $\Omega$ and $\pO$. This is the well known Lions trace theorem (see, e.g., \cite[Theorem 7.39]{adams03} and 
Theorem~9.4 of Chapter~1 in \cite{lions72}).

\begin{theorem}[Lions]\label{LMtracetheorem}
Let $\Phi\in\C^{\infty}(\Omega)$ and let $\gamma\colon\C^\infty(\Omega)\to\C^{\infty}(\pO)$ be the trace map $\gamma(\Phi)=\Phi\bigr|_{\pO}$. 
There is a unique continuous extension of the trace map such that
\begin{enumerate}
\item $\gamma\colon\H^{k}(\Omega)\to\H^{k-1/2}(\pO)$\,, $k > 1/2$\;.\\
\item The map is surjective\;.\\
\end{enumerate}
\end{theorem}


\subsection{A class of self-adjoint extensions of the Laplace-Beltrami operator}

We recall here some results from \cite{ILP13} that describe a large class of self-adjoint extensions of the Laplace-Beltrami operator. 
The extensions are parameterized in terms of suitable unitaries on the boundary Hilbert space. This class is constructed in 
terms of a family of closed, semi-bounded quadratic forms via the representation theorem (cf., Theorem \ref{fundteo}). 
Before introducing this family we shall need some definitions.

\begin{definition}\label{DefGap}
Let $U\colon  \L^2(\pO)\to\L^2(\pO)$ be unitary and denote by $\sigma(U)$ its spectrum. The unitary $U$ on the boundary
\textbf{has spectral gap at $-1$} if one of the following conditions hold:
\begin{enumerate}
\item $\1+U$ is invertible.
\item $-1\in\sigma(U)$ and $-1$ is not an accumulation point of $\sigma(U)$.
\end{enumerate}
The eigenspace associated to the eigenvalue $-1$ is denoted by $W$. The corresponding orthogonal projections will be written as
$P_W$ and $P_{W^\perp}=\1-P_W$. 
\end{definition}

\begin{definition}\label{partialCayley}
Let $U$ be a unitary operator acting on $\L^2(\pO)$ with spectral gap at $-1$\,. 
The \textbf{partial Cayley transform} $A_U\colon\L^2(\pO)\to \L^2(\pO)$ is
the operator
\[
A_U:=\mathbf{i}\,P_{W^\bot} (U-\mathbb{I}) (U+\mathbb{I})^{-1}\;.
\]
\end{definition}

\begin{definition}\label{def:admissible}
Let $U$ be a unitary with spectral gap at $-1$\,. The unitary is said to be \textbf{admissible} if the partial Cayley transform
$A_U$ leaves the subspace $\H^{1/2}(\pO)$ invariant and is continuous with respect to the Sobolev norm of order $1/2$\,, i.e.,
$$\norm{A_U\varphi}_{\H^{1/2}(\pO)}\leq K \norm{\varphi}_{\H^{1/2}(\pO)}\;.$$
\end{definition}

\begin{example}
Consider a manifold with boundary given by the unit circle, i.e., $\partial\Omega=S^1$\,,
and define the unitary $(U_\beta\varphi)(z):=e^{i\beta(z)}\,\varphi(z)$\,, $\varphi\in \L^2(S^1)$\,.
If $\beta\in \L^2(S^1)$ and $\ran\beta\subset\{\pi\}\cup [0,\pi-\delta]\cup [\pi+\delta,2\pi)$\,, for some $\delta >0$\,, then $U_\beta$ has gap
at $-1$\,. If, in addition, $\beta\in C^\infty(S^1)$\,, then $U_\beta$ is admissible.
\end{example}

\begin{definition}\label{DefQU}
Let $U$ be a unitary with spectral gap at $-1$\,, $A_U$ the corresponding partial Cayley transform and $\gamma$
the trace map considered in Theorem~\ref{LMtracetheorem}.
The Hermitean quadratic form $Q_U$ associated to the unitary $U$ is defined by
$$Q_U(\Phi,\Psi)=\scalar{\d\Phi}{\d\Psi}_{\Lambda^1}-\scalarb{\gamma(\Phi)}{A_U\gamma(\Phi)}$$
on the domain
$$\D_U=\bigl\{ \Phi\in\H^1(\Omega)\bigr|   \; P_W\gamma(\Phi)=0 \bigr\}\;.$$ 
Here $\scalar{\cdot}{\cdot}_{\Lambda^1}$ 
stands for the canonical scalar product among one-forms on the manifold $\Omega$\,.
\end{definition}

It is worth to mention the reasons behind Definitions~\ref{DefGap} and \ref{def:admissible}. 
The spectral gap condition ensures that the partial Cayley transform becomes a bounded, self-adjoint operator on the subspace
$W^\perp$ and this guarantees that the quadratic form $Q_U$ is lower semi-bounded. Notice that we are dealing with 
unbounded quadratic forms and thus they are not continuous mappings of the Hilbert space. The admissibility condition is 
an analytic requirement to ensure that $Q_U$ is a closable quadratic form. 

In the next theorem we give a class of self-adjoint extensions of the minimal Laplacian
operator $\Delta_{\mathrm{min}}$ on the domain $\mathcal{H}^2_0(\Omega)$.  
We refer to \cite[Section~4]{ILP13} for a complete proof and additional motivation.
All the extensions are labeled by suitable unitaries $U$ at the boundary.

\begin{theorem}\label{char_U}
Let $U\colon\L^2(\pO)\to\L^2(\pO)$ be an admissible unitary operator with spectral gap at $-1$. 
Then the quadratic form $Q_U$ of Definition~\ref{DefQU} is semi-bounded from below and closable. Its closure 
is represented by a semi-bounded, self-adjoint extension of the minimal Laplacian $-\Delta_{\mathrm{min}}$. 
\end{theorem}


\subsection{Unitaries at the boundary and $G$-invariance}\label{subsec:6-2}
We will use next the results of Section~\ref{sec:InvariantQF} 
to give necessary and sufficient conditions on the characterization of the unitary $U$ in order
that the corresponding quadratic form $(Q_U,\D_U)$ is $G$-invariant. In particular, 
from Theorem~\ref{QAinvariant} we conclude 
that the self-adjoint operator representing its closure will also be $G$-invariant.

We need to consider first the quadratic form corresponding to the Neumann extension of the Laplace-Beltrami
operator: 
\[
Q_N(\Phi)=\|d\Phi\|^2 \;,\quad\quad \Phi\in\D_N=\mathcal{H}^1 (\Omega)\;.
\]
We will also call it Neumann quadratic form. Note that it corresponds to the quadratic form 
$Q_U$ of the previous section with admissible unitary $U = \mathbb{I}$. Moreover,
$U$ has spectral gap at $-1$ and for the corresponding orthogonal projection we have 
$P_{W}=\1$, hence $A_U=0$.

Let $G$ be a Lie group and $V\colon G \to \mathcal{U}(\L^2(\Omega))$ be a continuous unitary representation of $G$, 
i.e., for any $\Phi\in \L^2(\Omega)$ the map
\[
 G\ni g\mapsto V(g)\Phi \quad 
\]
is continuous in the $\L^2$-norm. We will assume that $Q_N$ is $G$-invariant, that is, 
$V(g) \mathcal{H}^1(\Omega) \subset \mathcal{H}^1(\Omega)$ 
and $Q_N(V(g) \Phi ) = Q_N(\Phi)$ for all $\Phi \in \mathcal{H}^1(\Omega)$ and $g\in G$. Then the following lemma 
shows that $V$ defines also a continuous unitary representation on $\mathcal{H}^{1}(\Omega)$ with its corresponding
Sobolev scalar product.

\begin{lemma}  
Let $V$ a strongly continuous unitary representation of the Lie group $G$ on $\L^2(\Omega)$ such that 
Neumann's quadratic form $Q_N$ is $G$-invariant.  Then $V$ leaves invariant the subspace $\mathcal{H}^{1}(\Omega)$ 
and defines a strongly continuous unitary representation on it.
\end{lemma}
\begin{proof}  
Since $V(g)$ is invertible it is enough to show that $V(g)$ is an isometry with respect to the 
Sobolev norm $|| \cdot ||_1$ (see also the proof of Theorem~\ref{thm:unit-scales}). But this is immediate
since $V$ is unitary on $\L^2(\Omega)$ and $Q_N$ is $G$-invariant.
This is trivial if we realize that $||\cdot ||_1^2 = ||\cdot ||_0^2 + Q_N(\cdot )$, then because of the $G$-invariance 
of $Q_N$ we get that $||V(g) \Phi ||_1 = || \Phi ||_1$ for all $\Phi \in \mathcal{H}^{1}(\Omega)$, $g \in G$.

Finally, to prove strong continuity on $\mathcal{H}^{1}(\Omega)$ use the invariance
\[
 \| V(g)\Phi \|_1 = \norm{\Phi}_1\;,\quad g\in G
\]
and a standard weak compactness argument.
\end{proof}

\begin{definition} \label{def:traceable}
The representation $V\colon G\to \L^2(\Omega)$ {\bf has a trace (or that it is traceable) along the boundary} 
$\partial \Omega$ if there exists another continuous, unitary representation $v \colon G\to \mathcal{U}( \L^2(\partial \Omega))$ 
such that 
\begin{equation}\label{VPhi=vphi}
\gamma( V(g) \Phi ) = v(g)  \gamma (\Phi)  \, ,
\end{equation} 
for all $\Phi \in \mathcal{H}^1(\Omega)$ and $g\in G$ or, in other words, that the following diagram is commutative:
$$
\begin{array}{ccc}   \mathcal{H}^{1}(\Omega) & \overset{V(g)}{\longrightarrow} & \mathcal{H}^{1}(\Omega) \\
\gamma \downarrow & & \downarrow \gamma \\
 \mathcal{H}^{1/2}(\partial\Omega) & \overset{v(g)}{\longrightarrow}& \mathcal{H}^{1/2}(\partial\Omega)
\end{array}
$$
We will call $v$ the trace of the representation $V$.
\end{definition}

Notice that if the representation $V$ is traceable, its trace $v$ is unique.  

Now we are able to prove the following theorem:

\begin{theorem}\label{repcommutation}
Let $G$ a Lie group, and $V\colon G\to \mathcal{U}(\L^2(\Omega))$ a traceable continuous, unitary representation of $G$ with unitary trace
$v\colon G\to\mathcal{U}(\L^2(\pO))$ along the boundary $\partial\Omega$. Denote by $(Q_U,\mathcal{D}_U)$ the closable and 
semi-bounded quadratic form of Definition~\ref{DefQU} with an admissible unitary $U\in\mathcal{U}(\L^2(\Omega))$ 
having spectral gap at $-1$. Assume that the corresponding 
Neumann quadratic form $Q_N$ is $G$-invariant. Then we have the following cases:
\begin{itemize}
 \item[(i)] If $[v (g)\,,\,U]=0$ for all $g\in G$, then $Q_U$ is $G$-invariant. Its closure is also $G$-invariant and the self-adjoint extension
            of the minimal Laplacian representing the closed quadratic form will also be $G$-invariant.
 \item[(ii)] Consider the decomposition of the boundary Hilbert space $\L^2(\pO)\cong W\oplus W^\perp$, where $W$ is the eigenspace associated
             to the eigenvalue $-1$ of $U$ and denote by $P_W$ the orthogonal projection onto $W$. If $Q_U$ is $G$-invariant and 
             $P_W\colon\mathcal{H}^{1/2}(\partial\Omega)\to \mathcal{H}^{1/2}(\partial\Omega)$ continuous, then $[v (g)\,,\,U]=0$
             for all $g\in G$.
\end{itemize}
\end{theorem}

\begin{proof}
(i) Assume first that $[v (g)\,,\,U]=0$, $g\in G$. To show that $Q_U$ is $G$-invariant we have to analyze first the block structure of 
$U$ and $v(g)$ with respect to the decomposition
$
 \L^2(\pO)\cong W\oplus W^\perp
$.
Since
\[
U\cong 
\begin{pmatrix}
-\1 & 0 \\ 0 & U_0
\end{pmatrix}
\quad\mathrm{and}\quad
v(g)\cong
\begin{pmatrix}
v_1(g) & v_2(g) \\ v_3(g) & v_4(g)
\end{pmatrix}
\]
the commutation relations imply that $v_2(g)(\1+U_0)=0$ and $[v_4 (g)\,,\,U_0]=0$. But since $U$ has spectral gap at $-1$, then $\1+U_0$ is invertible
on $W^\perp$ and we must have $v_2(g)=0$. The unitarity of $v(g)$ implies $v_3(g)=0$ and $v(g)$ has block diagonal structure.  

By assumption $Q_N$ is $G$-invariant, so it is enough to show that the boundary quadratic form
\[
B(\Phi):=\scalarb{\gamma(\Phi)}{A_U\gamma(\Phi)}\;,\quad \Phi\in\D_U
\]
is also $G$-invariant, i.e., $V(g)\D_U\subset \D_U$ and $B(V(g)\Phi)=B(\Phi)$, $\Phi\in\D_U$.
To show the first inclusion, note that for any $\Phi\in\D_U$ we have 
\[
P_W\gamma(V(g)\Phi)=P_Wv(g) \gamma(\Phi)=v(g)P_W\gamma(\Phi)=0\;,
\]
where we have used that $V$ is traceable along $\pO$, $v(g)$ has diagonal block structure and 
$P_W\cong 
\begin{pmatrix}
\1 & 0 \\ 0 & 0
\end{pmatrix}$.

Finally, for any $g\in G$ and $\Phi\in\D_U$ we check 
\begin{align*}
B(V(g)\Phi)&=
\scalarbw{v_4(g)\gamma(\Phi)}{A_Uv_4 (g)\gamma(\Phi)}=\scalarbw{v_4 (g)\gamma(\Phi)}{v_4 (g)A_U\gamma(\Phi)}\\
&=\scalarbw{\gamma(\Phi)}{A_U\gamma(\Phi)}
=B(\Phi)\;.
\end{align*}
Note that all scalar products refer to $W^\perp$ and that for the last equation we used $v_4(g)\in\{U_0\}'$ iff
$v_4(g)\in\{A_U\}'$ and, again, all commutants are taken with respect to $W^\perp$ (cf., Section~\ref{sec:reduction}).
Since $Q_U$ is $G$-invariant and closable and the representation $V$ unitary it is straightforward to show that its
closure is also $G$-invariant (see, e.g., Theorem~VI.1.17 in \cite{kato95}). By Theorem~\ref{QAinvariant} it follows
that the self-adjoint extension of the minimal Laplacian representing the closed quadratic form will also be $G$-invariant.

(ii) By assumption we have that
\[
 P_W\colon\H^{1/2}(\pO)\to\H^{1/2}(\pO)
\]
is continuous in the fractional Sobolev norm and, therefore, $\gamma(\D_U)$ is dense in $W^\perp$. Since $Q_U$ is $G$-invariant
we have $V(g)\D_U\subset \D_U$, hence for any $\Phi\in\D_U$
\[
 0= P_W \gamma(V(g)\Phi) = P_W v(g) \gamma(\Phi) = P_W v(g) P_W^\perp \gamma(\Phi) \;.
\]
From the density of  $\gamma(\D_U)$ in $W^\perp$ we conclude that 
$v(g)\cong
\begin{pmatrix}
v_1(g) & 0 \\ 0 & v_4(g)
\end{pmatrix}
$
and we only need to show $[v_4 (g)\,,\,U_0]=0$ on $W^\perp$. But this follows from the $G$-invariance
of the boundary quadratic form $B$ and, again, the density of  $\gamma(\D_U)$ in $W^\perp$.
\end{proof}

We can now consider the following immediate consequences. 
\begin{corollary}\label{cor:1}
Let $U\in\mathcal{U}(\L^2(\pO))$ be admissible and such that $\1+U$ is invertible. Then 
$Q_U$ on $\D_U$ is $G$-invariant iff $[v (g)\,,\,U]=0$ for all $g\in G$.
\end{corollary}
\begin{proof}
We only need to show that the $G$-invariance implies that the unitaries on the boundary commute.
Note that by assumption $-1\not\in\sigma(U)$ and, with the notation in the preceding theorem, we have
that $W=\{0\}$. Therefore $\D_U=\H^1(\Omega)$ and the corresponding trace gives 
$\gamma(D_U)=\H^{1/2}(\pO)$ which is dense in $\L^2(\pO)$. The rest of the reasoning is litteraly as
in the proof of part~(ii) of Theorem~\ref{repcommutation}.
\end{proof}

We conlcude giving a characterization of $G$-invariance of the quadratic form $Q_U$ 
that uses the point spectrum of the unitary $U$. Recall that $\lambda$ is in the point spectrum
of an operator $T$ if $(\lambda-T)$ is not injective, i.e., $\lambda$ is an eigenvalue 
of $T$. We denote the set of all eigenvalues by $\sigma_p(T)\subset\sigma(T)$.

\begin{lemma}\label{lem:2}
Consider a unitary $U\in\mathcal{U}(\L^2(\pO))$ with spectral gap at $-1$. Assume that 
$\H^{1/2}(\pO)$ is invariant for $U$ and that its restriction
\[
 U_+:=U|\H^{1/2}(\pO)
\]
is continuous with respect to the Sobolev $1/2$-norm. If $U_+$ has only point spectrum, 
i.e., $\sigma(U_+)=\sigma_p(U_+)$, then $U$ is admissible, i.e., the partial Cayley transform
$A_U$ leaves the fractional Sobolev space $\H^{1/2}(\pO)$ invariant and is continuous with respect to the Sobolev $1/2$-norm.
The orthogonal projection $P_W$ leaves also the Sobolev space $\H^{1/2}(\pO)$ invariant and is continuous 
with respecto to the $1/2$-norm.
\end{lemma}
\begin{proof}
Note first that
\[
 \sigma(U_+)\subset\sigma_p(U_+)\subset\sigma_p(U)\subset \sigma(U)
\]
and, therefore, if $U$ has spectral gap at $-1$, then $U_+$ has also spectral gap at $-1$.
Then by Cauchy-Riesz functional calculus (cf., \cite[Chapter~VII]{DunfordI}) we have
that
\[
 A_U=\frac{1}{2\pi i}\int_{c_1}\,i\;\frac{\lambda-1}{\lambda+1}\left(\lambda-U\right)^{-1}\; d\lambda
 \quad\mathrm{and}\quad
 P_W=\frac{1}{2\pi i}\int_{c_2}\left(\lambda-U\right)^{-1}\; d\lambda\;,
\]
where $c_1, c_2$ are closed, simple and positively oriented curves. The curve $c_1$ encloses $\sigma(T)\setminus\{-1\}$ and
$c_2$ encloses only $\{-1\}$. Note that the the gap condition is essential here.
It is clear that both $A_U$ and $P_W$ are bounded operators in $\H^{1/2}(\pO)$ as required.
\end{proof}

\begin{corollary}\label{cor:2}
Consider a unitary $U\in\mathcal{U}(\L^2(\pO))$ with spectral gap at $-1$.
Assume that $\H^{1/2}(\pO)$ is invariant for $U$, that its restriction
$
 U_+:=U|\H^{1/2}(\pO)
$
is continuous with respect to the Sobolev $1/2$-norm and that
$\sigma(U_+)=\sigma_p(U_+)$. Then 
$Q_U$ on $\D_U$ is $G$-invariant iff $[v (g)\,,\,U]=0$ for all $g\in G$.
\end{corollary}
\begin{proof}
 By the preceding lemma we have that $U$ is admissible and that $P_W$ is continuous
in the $1/2$-norm. The statement follows then directly from Theorem~\ref{repcommutation}.
\end{proof}

In the final section we will present examples with unitaries which satisfy the conditions mentioned in the statements above.

\subsection{Groups acting by isometries}

We will discuss now the important instance when the unitary representation is determined by an action of the group $G$ on $\Omega$ by isometries.
Thus, assume that the group $G$ acts smoothly by isometries on the Riemannian manifold $(\Omega,\pO,\eta)$. 
Any $g\in G$ specifies a diffeomorphism $g:\Omega\to\Omega$ that we will denote with the same symbol for simplicity of notation. 
Moreover, we have that $$g^*\eta=\eta\;,$$ where $g^*$ stands for the pull-back by the diffeomorphism $g$. 
These diffeomorphisms restrict to isometric diffeomorphisms on the Riemannian manifold at the boundary $(\pO,\partial\eta)$
(see, e.g., \cite[Lemma 8.2.4]{marsden01}), 
$$
(g|_{_{\pO}})^*\partial\eta=\partial\eta\;.
$$ 
These isometric actions of the group $G$ induce unitary representations of the group on $\Omega$ and $\pO$\,.
In fact, consider the following representations:

\begin{equation*}
V\colon G\to\mathcal{U}(\L^2(\Omega))\;, \qquad V(g)\Phi=(g^{-1})^*\Phi\quad \Phi\in\L^2(\Omega)\;.
\end{equation*}
\begin{equation*}
v\colon G\to\mathcal{U}(\L^2(\pO))\;, \qquad
v (g)\varphi=(g|_{_{\pO}}^{-1})^*\varphi\quad \varphi\in\L^2(\pO)\;.
\end{equation*}
Then a simple computation shows that,
$$
\scalar{V(g^{-1})\Phi}{V(g^{-1})\Psi} =\scalar{\Phi}{\Psi}\;,
$$
where we have used the change of variables formula and the fact that isometric diffeomorphisms preserve the Riemannian volume, 
i.e., $g^*\negthinspace\d\mu_\eta=\d\mu_{\eta}$\,. The result for the boundary is proved similarly. 
The induced actions are related with the trace map as
$$
\gamma(V(g)\Phi)=v (g)\gamma(\Phi)\quad g\in G, \Phi\in\H^1(\Omega)\;,
$$
and therefore the unitary representation $V$ is traceable along the boundary of $\Omega$ with trace $v$.

Moreover we have that the quadratic form $Q_N$ is $G$-invariant.

\begin{proposition}\label{prop:dphiinvariant}
Let $G$ be a Lie group that acts by isometric diffeomorphisms on the Riemannian manifold 
$(\Omega,\pO,\eta)$ and let $V:G\to \mathcal{U}(\L^2(\Omega))$ be the associated unitary representation. 
Then, Neumann's quadratic form $Q_N (\Phi) = \scalar{\d\Phi}{\d\Phi}$ with domain $\H^1(\Omega)$ is $G$-invariant.
\end{proposition}

\begin{proof}
First notice that the pull-back of a diffeomorphism commutes with the action of the exterior differential. 
Then we have that $$\d(V(g^{-1})\Phi)=\d(g^*\Phi)=g^*\negthinspace\d\Phi\;.$$ 
Hence
\begin{subequations}\label{dphiinvariant}
\begin{align}
\scalar{\d(V(g^{-1})\Phi)}{\d(V(g^{-1})\Psi)}&=\int_\Omega \eta^{-1}(g^*\negthinspace\d\Phi,g^*\negthinspace\d\Psi)\d\mu_\eta \nonumber\\
&=\int_\Omega g^*\negthickspace\left( \eta^{-1}(\d\Phi,\d\Psi) \right)g^*\negthinspace\d\mu_\eta \nonumber\\
&=\int_{g\Omega}\eta^{-1}(\d\Phi,\d\Psi)\d\mu_\eta \nonumber\\
&=\scalar{\d\Phi}{\d\Psi}\;,
\end{align}
\end{subequations}
where in the second inequality we have used that $g:\Omega\to\Omega$ is an isometry and therefore 
$$\eta^{-1}(g^*\negthinspace\d\Phi,g^*\negthinspace\d\Psi)=g^*\negthinspace \eta^{-1}(g^*\negthinspace\d\Phi,g^*\negthinspace\d\Psi)
=g^*\negthickspace\left( \eta^{-1}(\d\Phi,\d\Psi) \right)\;.$$ 
The equations \eqref{dphiinvariant} guaranty also that $V(g)\H^1(\Omega)=\H^1(\Omega)$ since $V(g)$
is a unitary operator in $\L^2(\Omega)$ and the norm 
$\sqrt{\norm{\operatorname{d}\cdot\;}_{\Lambda^1}^2+\norm{\cdot}^2}$ is equivalent to the Sobolev norm of order 1.
\end{proof}

Before making explicit the previous structures in concrete examples we summarize the the main ideas in this section as follows: 
given a group acting by isometric diffeomorphisms on a Riemannian manifold, then any operator at the boundary, 
that satisfies the conditions in Definition~\ref{DefGap} and Definition~\ref{def:admissible}, and that verifies the commutation relations of 
Theorem~\ref{repcommutation}~(i) describes a $G$-invariant quadratic form. The closure of this quadratic form characterizes uniquely a 
$G$-invariant self-adjoint extension of the Laplace-Beltrami operator (cf.~Theorems~\ref{fundteo} and \ref{QAinvariant}).

\section{Examples}
\label{sec:SymmetryExamples}

In this section we introduce two particular examples of $G$-invariant quadratic forms. In the first example we are considering a 
situation where the symmetry group is a finite, discrete group. In the second one we consider $G$ to be a compact Lie group.

\begin{example}
Let $\Omega$ be the cylinder $[-1,1]\times[-1,1]/\negthickspace\sim$\,, where $\sim$ is the equivalence relation $(x,1)\sim(x,-1)$\,. 
The boundary $\pO$ is the disjoint union of the two circles $\Gamma_1=\bigl\{\{-1\}\times[-1,1]/\negthickspace\sim\bigr\}$ and
$\Gamma_2=\bigl\{\{1\}\times[-1,1]/\negthickspace\sim\bigr\}$\,, (see Figure \ref{fig:cilindro}). Let $\eta$ be the euclidean metric on $\Omega$. 
Now let $G=\mathbb{Z}_2$=\{e,f\} be the discrete, abelian group of two elements and consider the following action in $\Omega$:
\begin{align*}
e:(x,y)&\to(x,y)\;,\\
f:(x,y)&\to(-x,y)\;.
\end{align*}
The induced action at the boundary is
\begin{align*}
e:(\pm1,y)\to(\pm1,y)\;,\\
f:(\pm1,y)\to(\mp1,y)\;.
\end{align*}
Clearly $G$ transforms $\Omega$ onto itself and preserves the boundary. Moreover, it is easy to check that $f^*\eta=\eta$\,.
\begin{figure}[h]
\centering
\includegraphics[width=12cm]{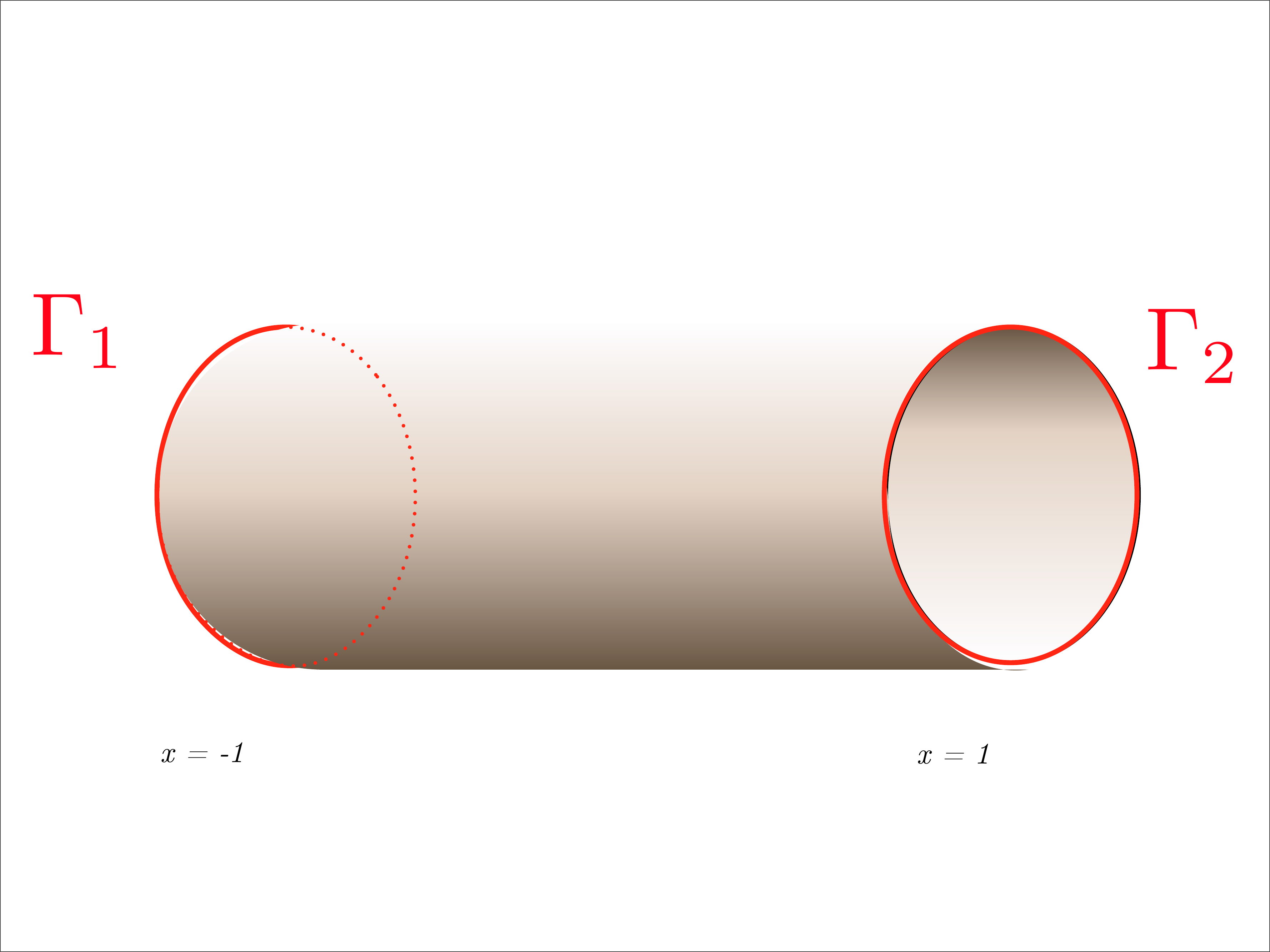}
\caption{}\label{fig:cilindro}
\end{figure}
Since the boundary $\pO$ consists of two disjoints manifolds $\Gamma_1$ and $\Gamma_2$\,, the Hilbert space of 
the boundary is $\L^2(\pO)=\L^2(\Gamma_1)\oplus\L^2(\Gamma_2)$. Any $\Phi\in\L^2(\pO)$
can be written as
$$\Phi=\begin{pmatrix}\Phi_1(y)\\ \Phi_2(y)\end{pmatrix}$$ 
with $\Phi_i\in\L^2(\Gamma_i)$\,. 
The only nontrivial action on $\L^2(\pO)$ is given by 
\[
v (f)\begin{pmatrix}\Phi_1(y)\\ \Phi_2(y)\end{pmatrix}
=\begin{pmatrix}0 & \mathbb{I}\\ \mathbb{I} & 0\end{pmatrix}\begin{pmatrix}\Phi_1(y)\\ \Phi_2(y)\end{pmatrix}\;.
\]
The set of unitary operators that describe the closable quadratic forms as defined in the previous section 
is given by suitable unitary operators 
$$U=\begin{pmatrix} U_{11} & U_{12} \\ U_{21} & U_{22} \end{pmatrix}\;,$$ 
with $U_{ij}=\L^2(\Gamma_j)\to\L^2(\Gamma_i)$. 
According to Theorem~\ref{repcommutation}~(i) the unitary operators commuting 
with $v (f)$ will lead to $G$-invariant quadratic forms. Imposing
$$\begin{pmatrix} 0 & \mathbb{I} \\ \mathbb{I} & 0\end{pmatrix}
\begin{pmatrix} U_{11} & U_{12} \\ U_{21} & U_{22} \end{pmatrix}=\begin{pmatrix} U_{11} & U_{12} \\ U_{21} & U_{22} 
\end{pmatrix}\begin{pmatrix} 0 & \mathbb{I} \\ \mathbb{I} & 0\end{pmatrix}\;,$$
we get the conditions
\begin{align*}
&U_{21}-U_{12}=0\;,\\
&U_{22}-U_{11}=0\;.
\end{align*}

Obviously there is a wide class of unitary operators, i.e., boundary conditions, that will be compatible with the symmetry group $G$. 
We will consider next two particular classes of boundary conditions. 
First, consider the following unitary operators 
\begin{equation}
U=\begin{bmatrix} e^{\mathrm{i}\beta_1}\mathbb{I}_1 & 0\\ 0 & e^{\mathrm{i}\beta_2}\mathbb{I}_2 \end{bmatrix}\;,
\end{equation}
where $\beta_i\in\C^\infty\left(S^1,[-\pi+\delta,\pi-\delta]\cup\{\pi\}\right)$ for some $\delta>0$. It is showed in 
\cite[Sections~3 and 5]{ILP13} that this class of 
unitary operators have spectral gap at -1 and are admissible (see also Subsection~\ref{subsec:6-2}). 
Moreover, this choice of unitary matrices corresponds to select Robin 
boundary conditions of the form:
\begin{equation}
\gamma\left(-\frac{\d\Phi}{\d x}\right)\biggr|_{\Gamma_1}=-\tan(\beta_1/2)\gamma(\Phi)\bigr|_{\Gamma_1}\quad;\quad 
\gamma\left(\frac{\d\Phi}{\d x}\right)\biggr|_{\Gamma_2}=-\tan(\beta_2/2)\gamma(\Phi)\bigr|_{\Gamma_2}\;.
\end{equation}
The $G$-invariance condition above imposes $\beta_1=\beta_2$. Notice that when $\beta_1\neq\beta_2$ 
we can obtain meaningful self-adjoint extensions of the Laplace-Beltrami operator that, however, will not be $G$-invariant.

We can also consider unitary operators of the form
\begin{equation}
U=\begin{bmatrix} 0 & e^{\mathrm{i}\alpha} \\ e^{-\mathrm{i}\alpha} & 0 \end{bmatrix}\;,
\end{equation}
where $\alpha\in \C^\infty (S^1,[0,2\pi])$\,. Again, it is proved in \cite{ILP13} that this class of unitary operators have spectral 
gap at $-1$ and are admissible.  
In this case the unitary matrix corresponds to select so-called quasi-periodic boundary conditions, 
cf., \cite{asorey83}, i.e.,
\begin{equation*}
  \gamma(\Phi)\bigr|_{\Gamma_1}=e^{i\alpha}\gamma(\Phi)\bigr|_{\Gamma_2}\;,
\end{equation*}
\begin{equation*}
  \gamma\left(-\frac{\d\Phi}{\d x}\right)\biggr|_{\Gamma_1}=e^{i\alpha}\gamma\left(\frac{\d\Phi}{\d x}\right)\biggr|_{\Gamma_2}\;.
\end{equation*}
The $G$-invariance condition imposes $e^{i\alpha}=e^{-i\alpha}$ and therefore among all the quasi-periodic conditions only the
periodic ones, $\alpha\equiv0$\,, are compatible with the $G$-invariance. 
\end{example}

\begin{example}\label{ex:upperhemisphere}
Let $\Omega$ be the unit, upper hemisphere centered at the origin. 
Its boundary $\pO$ is the unit circle on the $z=0$ plane. 
Let $\eta$ be the induced Riemannian metric from the euclidean metric in $\mathbb{R}^3$\,. Consider the compact Lie group 
$G=SO(2)$ that acts by rotation around the $z$-axis. 
If we use polar coordinates on the horizontal plane, then the boundary $\pO$ is isomorphic to the interval $[0,2\pi]$ with 
the two endpoints identified. 
We denote by $\theta$ the coordinate parameterizing the boundary and the boundary Hilbert space is $\L^2(S^1)$\,.

Let $\varphi\in\H^{1/2}(\pO)$ and consider the action on the boundary by a group element 
$g_\alpha\in G$, $\alpha\in [0,2\pi]$, given by
$$v (g^{-1}_\alpha)\varphi(\theta)=\varphi(\theta+\alpha)\;.$$
To analyze what are the possible unitary operators that lead to $G$-invariant quadratic forms it is convenient to use the Fourier 
series expansions of the elements in $\L^2(\pO)$\,. Let $\varphi\in\L^2(\pO)$\,, then
$$\varphi(\theta)=\sum_{n\in\mathbb{Z}}\hat{\varphi}_ne^{\mathrm{i}n\theta}\;,$$
where the coefficients of the expansion are given by $$\hat{\varphi}_n=\frac{1}{2\pi}\int_0^{2\pi}\varphi(\theta)e^{-\mathrm{i}n\theta}\d\theta\;.$$
We can therefore consider the induced action of the group $G$ as a unitary operator on ${\ell}_2$\,, the Hilbert space of square summable sequences. 
In fact we have that:
\begin{align*}
\widehat{(v (g^{-1}_\alpha)\varphi)}_n&
=\frac{1}{2\pi}\int_0^{2\pi}\varphi(\theta+\alpha)e^{-\mathrm{i}n\theta}\d\theta\\
&=\sum_{m\in\mathbb{Z}}\hat{\varphi}_me^{\mathrm{i}m\alpha}\int_0^{2\pi}\frac{e^{\mathrm{i}(m-n)\theta}}{2\pi}\d\theta
 =e^{\mathrm{i}n\alpha}\hat{\varphi}_n\;.
\end{align*}
This shows that the induced action of the group $G$ is a unitary operator in $\mathcal{U}(\ell_2)$ that acts 
diagonally in the Fourier series expansion.
More concretely, we can represent it as $\widehat{v (g^{-1}_\alpha)}_{nm}=e^{\mathrm{i}n\alpha}\delta_{nm}\,$\,.
From all the possible unitary operators acting on the Hilbert space of the boundary, only those whose representation in $\ell_2$
commutes with the above operator will lead to $G$-invariant quadratic forms (cf., Theorem~\ref{repcommutation}~(i)). 
Since $\widehat{v (g^{-1}_\alpha)}$ acts
diagonally on $\ell_2$ it is clear that only operators of the form $\hat{U}_{nm}=e^{\mathrm{i}\beta_n}\delta_{nm}$\,, 
$\{\beta_n\}_n\subset\mathbb{R}$\,, will lead to $G$-invariant quadratic forms.

As a particular case we can consider that all the parameters are equal, i.e., $\beta_n=\beta$, $n\in\mathbb{Z}$\,.
In this case it is clear that $(\widehat{U\varphi})_n=e^{\mathrm{i}\beta}\varphi_n$\,, which gives the following 
admissible unitary with spectral gap at $-1$:
$$U\varphi=e^{\mathrm{i}\beta}\varphi\;.$$
This shows that the unique Robin boundary conditions compatible with the $SO(2)$-invariance are those that are defined with a 
constant parameter along the boundary, i.e., 
\begin{equation}
\gamma\left(\frac{\d\Phi}{\d \vec{n}}\right)=-\tan(\beta/2)\gamma(\Phi)\;,\quad\beta\in[0,2\pi]\;,
\end{equation}
where $\vec{n}$ stands for normal vector field pointing outwards to the boundary.
\end{example}


\end{document}